%% file: TCOM_draft_v1.tex
\documentclass[11pt,draftclsnofoot,onecolumn,journal,letterpaper]{IEEEtran}
\usepackage{cite}
\usepackage{amsmath,amssymb,amsfonts}
\usepackage{algorithmic}
\usepackage{graphicx}
\usepackage{textcomp}
\usepackage{xcolor}
\def\BibTeX{{\rm B\kern-.05em{\sc i\kern-.025em b}\kern-.08em
    T\kern-.1667em\lower.7ex\hbox{E}\kern-.125emX}}

\usepackage{graphicx}
\usepackage{epstopdf}
\usepackage{amsmath,amssymb,amsthm,mathrsfs,amsfonts,dsfont}
\usepackage{epsfig}
\usepackage[ruled,linesnumbered]{algorithm2e}
\usepackage{color}
\usepackage{subfigure}
\usepackage{cite}
\usepackage{diagbox}
\usepackage{multirow}
\usepackage{float}
\usepackage{stfloats}
\usepackage{enumitem}
\usepackage{cases}
\usepackage{setspace}
\usepackage{adjustbox,lipsum}
\usepackage{comment}
\usepackage{lipsum}
\usepackage{url}
\include{Backup/com}

\graphicspath{{figures/}}
\allowdisplaybreaks

\newcommand{\mypara}[1]{{\smallskip \noindent \bf #1}\hspace{0.1in}}
\newcommand{\ssf}[1]{\textrm{$\sf{#1}$}{}}
\newcommand{\real}{\text{Re}}

\newtheorem{theorem}{Theorem}

\newtheorem{lemma}{Lemma}
\newtheorem{assumption}{Assumption}
\newtheorem{coro}[theorem]{Corollary}

\newtheorem{remark}{Remark}

\DeclareMathOperator*{\argmin}{arg\,min}

\newcommand{\vect}[1]{\mathbf{#1}}
\newcommand{\avgvect}[1]{\mathbf{\overline{#1}}}
\newcommand{\expt}{\mathbb{E}}
\newcommand{\variance}{\mathbb{V}\text{ar}}
\newcommand{\norm}[1]{\left \| #1 \right \|}
\newcommand{\squab}[1]{\left [ #1 \right ]}
\newcommand{\dotp}[2]{\left \langle #1, #2 \right \rangle}

\ifodd 0

\else

\fi

\ifodd 0
\newcommand{\congc}[1]{{\color{red}(Cong: #1)}}
\else
\newcommand{\congc}[1]{}
\fi

\ifodd 0
\newcommand{\zixiang}[1]{{\color{blue}(Zixiang: #1)}}
\else
\newcommand{\zixiang}[1]{}
\fi

\ifodd 0
\newcommand{\zixiangb}[1]{{\color{blue}#1}}
\else
\newcommand{\zixiangb}[1]{#1}
\fi

\begin{document}


\title{Random Orthogonalization for Federated Learning in Massive MIMO Systems}



\author{Xizixiang~Wei \qquad Cong~Shen \qquad Jing~Yang \qquad H.~Vincent~Poor
\thanks{A preliminary version of this work has been presented at the 2022 IEEE International Conference on Communications \cite{wei2022icc}.}
\thanks{Xizixiang Wei and Cong Shen are with the Charles L. Brown Department of Electrical and Computer Engineering, University of Virginia, USA. (E-mail: \texttt{\{xw8cw,cong\}@virginia.edu}.)

Jing Yang is with the Department of Electrical Engineering, The Pennsylvania State University, USA. (E-mail: \texttt{yangjing@psu.edu}.)

H. Vincent Poor is with the Department of Electrical and Computer Engineering, Princeton University, USA. (E-mail: \texttt{poor@princeton.edu}.)}
}

\maketitle

\begin{abstract}

We propose a novel communication design, termed \emph{random orthogonalization}, for federated learning (FL) in a massive multiple-input and multiple-output (MIMO) wireless system. The key novelty of random orthogonalization comes from the tight coupling of FL and two unique characteristics of massive MIMO -- channel hardening and favorable propagation. As a result, random orthogonalization can achieve natural over-the-air model aggregation without requiring transmitter side channel state information (CSI) for the uplink phase of FL, while significantly reducing the channel estimation overhead at the receiver. We extend this principle to the downlink communication phase and develop a simple but highly effective model broadcast method for FL. We also relax the massive MIMO assumption by proposing an enhanced random orthogonalization design for both uplink and downlink FL communications, that does not rely on channel hardening or favorable propagation. Theoretical analyses with respect to both communication and machine learning performance are carried out. In particular, an explicit relationship among the convergence rate, the number of clients, and the number of antennas is established. Experimental results validate the effectiveness and efficiency of random orthogonalization for FL in massive MIMO.


\end{abstract}


\section{Introduction}
\label{sec:intro}

Machine learning (ML) model communication is widely considered as one of the primary bottlenecks for federated learning (FL) \cite{mcmahan2017fl,konecny2016fl,chen2021communication}. This is because an FL task consists of multiple learning rounds, each of which requires uplink and downlink model exchanges between clients and the server. The limited communication resources in both uplink and downlink, combined with the detrimental effects from channel fading, noise, and interference, severely impact the \emph{scalability} (in terms of the number of participating clients) of FL in a wireless communication system.

One promising technique to tackle the scalability problem of FL over wireless communications is over-the-air computation (also known as {AirComp}); see \cite{niknam2020federated} and the references therein.  Instead of the standard approach of decoding the individual local models of each client and then aggregating, AirComp allows multiple clients to transmit uplink signals in a superpositioned fashion, and decodes the average global model directly at the FL server. In order to achieve this goal, a common approach is to ``invert'' the fading channel at each transmitter \cite{zhu2019broadband,cao2020tpc}, so that the sum model can be obtained at the server. AirComp has attracted considerable interest and a detailed literature review can be found in Section~\ref{sec:related}.  

However, much of the existing work on AirComp has several limitations. First, these methods often require channel state information at the transmitter (CSIT) for each individual client. The process of enabling individual CSIT is complicated -- in a frequency division duplex (FDD) system, this involves the receiver estimating the channels and then sending back the estimates to the transmitters; in a time division duplex (TDD) system, one can benefit from channel reciprocity \cite{TV:05,G:05}, but there is still a need for an independent pilot for each client. In both cases, practical mechanisms to obtain individual CSIT do not scale with the number of clients. In addition, the precision of CSIT is often worse than that of channel state information at the receiver (CSIR).  Second, most AirComp approaches in the literature require a channel inversion-type power control, which is well known to ``blow up'' when at least one of the channels is experiencing deep fading \cite{TV:05}. 
\zixiangb{Third, AirComp approaches focus on improving the scalability and efficiency of the uplink communication phase in FL. How to address these challenges in the downlink communication phase remains underdeveloped.} 

Another important limitation is that the AirComp solution does not naturally extend to multiple-input and multiple-output (MIMO) systems where the uplink and downlink channels become vectors. Compared with the studies in scalar channels, there are only a few recent papers that explore the potential of MIMO for wireless FL. MIMO beamforming design to optimize FL has been studied in \cite{yang2020federated,elbir2020federated}. Coding, quantization, and compressive sensing over MIMO channels for FL have been studied in \cite{huang2020physical,jeon2020compressive}. Nevertheless, none of these works tightly incorporates the unique properties of MIMO to the FL communication design. On the other hand, if we ignore the unique characteristics of FL, MIMO can also be utilized in a straightforward manner. \zixiangb{In the uplink phase, we can use conventional MIMO estimators such as zero-forcing (ZF) or minimum mean square error (MMSE) to estimate each local model, and then compute the global model. In the downlink phase, we can design MIMO precoders to broadcast the global model.} However, \zixiangb{these approaches incur a large channel estimation overhead}, especially when the channels have high dimensions. Moreover, matrix inversions in the ZF or MMSE estimators \zixiangb{and the optimization algorithms for the precoding design} are computationally demanding, in particular for massive MIMO. This increases the complexity and latency of the overall system. \zixiangb{In addition}, decoding individual local models also makes it easier for the server to sketch the data distribution of the clients, leading to potential privacy leakage.

This paper aims at designing simple-yet-effective FL communication methods that can efficiently address the scalability challenge in FL for both uplink and downlink phases. The novelty comes from a tight integration of MIMO and FL -- our design explicitly utilizes the characteristics of both components. To illustrate the key idea, we start with  \emph{massive} MIMO where the base station (BS) has a large number of antennas. The proposed framework only requires the BS to estimate a \emph{summation channel}, which significantly alleviates the burden on channel estimation\footnote{For example, a single pilot can be used by all clients as long as it is sent synchronously, regardless of the number of clients that participate in the current FL round.}. Moreover, our approach is agnostic to the number of clients, and thus improves the scalability of FL. By leveraging the unique {channel hardening} and {favorable propagation} properties of massive MIMO, the proposed principle, termed \emph{random orthogonalization}, allows the BS to directly compute the global model via a simple linear projection operation, hence achieving extremely low complexity and low latency in the uplink communication phase.  We then extend the random orthogonalization design to the downlink communication phase, which leads to a simple but highly effective model broadcast method for FL. As the random orthogonalization designs rely on channel hardening and favorable propagation to eliminate the interference, which do not always hold in practice (e.g., when the number of antennas is small), we further propose an \emph{enhanced random orthogonalization} design for both uplink and downlink FL communications, that leverages \emph{channel echos} to compensate for the lack of channel hardening and favorable propagation. The enhanced random orthogonalization design thus can be applied to a general MIMO system. To analyze the performances of random orthogonalization, we derive the Cramer-Rao lower bounds (CRLBs) of the average model estimation errors as a theoretical benchmark. Moreover, taking both interference and noise into consideration, a novel convergence bound of FL is derived for the proposed methods over massive MIMO channels. Notably, we establish an explicit relationship among the convergence rate, the number of clients, and the number of antennas, which provides practical design guidance for wireless FL. Extensive numerical results validate the effectiveness and efficiency of the proposed random orthogonalization principle in a variety of FL and MIMO settings.

\if{0}


This paper aims at designing simple-yet-effective FL communication methods that enable over-the-air computation for both uplink and downlink phases. To address the scalability challenge in FL, we explore another design degree of freedom (d.o.f.) in modern wireless systems: \emph{massive MIMO}. The proposed framework only requires the base station (BS) to estimate a \emph{summation channel}, which significantly alleviates the burden on uplink channel estimation in FL. Moreover, this approach is agnostic to the number of clients, making it attractive for the scalability of FL. By tightly integrating the channel hardening and favorable propagation properties of massive MIMO, the proposed principle, termed \emph{random orthogonalization}, allows the BS to directly compute the global model via a simple linear projection operation, thus achieving extremely low complexity and low latency for the uplink communication phase of FL.  
We then extend the random orthogonalization design to the downlink communication phase, which leads to a simple but highly effective model broadcast method for FL. 
As the vanilla random orthogonalization designs rely on channel hardening and favorable propagation to eliminate the interference, which does not always hold in practice (e.g., when the number of antennas is small), we further propose an enhanced random orthogonalization design for both uplink and downlink FL communications, that leverages \emph{channel echos} to compensate for the lack of channel hardening and favorable propagation. 
To analyze the performances of random orthogonalization, we derive the Cramer-Rao lower bounds (CRLBs) of the average model estimation as a theoretical benchmark. Moreover, taking both interference and noise into consideration, a novel convergence bound of FL is derived for the proposed method over massive MIMO channels. Notably, we establish an explicit relationship among the convergence rate, the number of clients $K$, and the number of antennas $M$, which provides practical design guidance for wireless FL. Extensive numerical results validate the effectiveness and efficiency of the proposed method.

The potential of MIMO for wireless FL has attracted interest recently.  
MIMO beamforming design to optimize FL has been studied in \cite{yang2020federated,elbir2020federated}. Coding, quantization, and compressive sensing over a (massive) MIMO channel for FL has been studied in \cite{huang2020physical,jeon2020compressive}. Nevertheless, none of these works tightly incorporates the unique properties of massive MIMO to the FL uplink communication design. On the other hand, massive MIMO can also be utilized in a straightforward manner, e.g., one can use traditional MIMO estimators such as zero-forcing (ZF) or minimum mean-square-error (MMSE) to estimate each local model, and then compute the global model. However, this heuristic approach requires large channel estimation overhead, especially in massive MIMO. Decoding individual local models also makes it easier for the server to sketch the data distribution of a client. Moreover, matrix inversion operations in ZF or MMSE estimators are computationally demanding, which increases the complexity and latency.

\mypara{Large channel estimation overhead.} To decode individual $\{w_{t+1}^{k,i}\}_{i = 1}^d$, traditional MIMO estimators usually require full CSIR, which consumes expensive time-frequency overhead for channel estimation when $M$ is large. Intuitively, estimating $\sum_{k\in[K]} w_{t+1}^{k,i}$ from the received signal does not require full CSIR, as summation result is a weaker estimate than figuring out each individual parameter. Therefore, traditional MIMO estimators may waste too much channel estimation overhead when directly applying to FL.

\mypara{Privacy concerns.} FL is proposed to protect users' privacy. However, decoding individual $\{w_k\}_{k=1}^K$ makes it easier for server to sketch data distribution of a specific user. To further address the privacy concern, we would like the server to obtain $\sum_{k\in[K]} w_{t+1}^{k,i}$ without a priori knowledge $\{w_{t+1}^{k,i}\}_{i = 1}^d$.

\mypara{High system latency.} Unlike AirComp, traditional MIMO estimators cannot directly generate a averaged local model and matrix inversion operations in ZF or MMSE estimator also require extra computations, which increase the system latency for server aggregation. 

\fi

The remainder of this paper is organized as follows. Related works are surveyed in Section~\ref{sec:related}. Section \ref{sec:model} introduces the FL pipeline and the wireless communication model. The proposed random orthogonalization principle is presented in Section \ref{sec:RO}, and then the enhanced design is proposed in Section~\ref{sec:enhanced}. 
Analyses of the CRLB as well as the FL model convergence are given in Section \ref{sec:analysis}. Experimental results are reported in Section \ref{sec:sim}, followed by the conclusions in Section \ref{sec:conc}.

\section{Related Works}
\label{sec:related}

\mypara{Improve FL communication efficiency.} The original Federated Averaging (\textsc{FedAvg}) algorithm \cite{mcmahan2017fl} reduces the communication overhead by only periodically averaging the local models. Theoretical understanding of the communication-computation tradeoff has been actively pursued and, depending on the underlying assumptions (e.g., independent and identically distributed (i.i.d.) or non-i.i.d. local datasets, convex or non-convex loss functions, gradient descent or stochastic gradient descent (SGD)), convergence analyses have been carried out \cite{stich2018local,wang2018cooperative,li2019convergence}. The approaches to reduce the payload size or communication frequency include sparsification \cite{thonglek2022sparse,oh2022fedvqcs} and quantization \cite{zhu2020one,amiri2020federated2,du2020high}. There are also efforts to improve resource allocation \cite{wen2021adaptive, chen2022federated,wang2021federated}. 

\mypara{AirComp for FL.}
As a special case of computing over multiple access channels \cite{nazer2007it}, AirComp \cite{zhu2019broadband,yang2020federated,amiri2020federated,cao2020tpc} leverages the signal superposition properties in a wireless multiple access channel to efficiently compute the average ML model. This technique has attracted considerable interest, as it can reduce the uplink communication cost to be (nearly) agnostic to the number of participating clients. Client scheduling and various power and computation resource allocation methods have been investigated \cite{chen2020convergencenew,xu2020client,sun2021dynamic,ma2021user, lee2021adaptive, wadu2021joint}. The assumption of full CSIT is relaxed in \cite{sery2020tsp} by only using the phase information of each individual channel. Convergence guarantees of Aircomp under different constraints are reported in \cite{lin2021deploying,aygun2022over, wan2021convergence,sery2021over,sun2022time}. 


\mypara{Communication design for FL in MIMO systems.} 
There are recent studies to optimize the communication efficiency and learning performance in MIMO systems for FL, including transmit power control \cite{vu2021energy,hamdi2021federated,vu2022joint}, data rate allocation \cite{vu2020cell}, compression \cite{jeon2020compressive,mu2022communication}, and learning rate optimization \cite{xu2021learning}. Several beamforming designs have been proposed to improve the performance of FL \cite{yang2020federated,lin2022distributed,zhong2021over,xia2020fast}. However, these methods require full CSIT and rely on complex optimization methods to design the beamformers, which is impractical in massive MIMO due to the high communication and computation cost. There is very limited study that relaxes the individual CSIT assumption in wireless FL over MIMO channels, with notable exceptions of \cite{amiri2021blind,tegin2021blind}. However, they only focus on the uplink communication phase.



\section{System Model}
\label{sec:model}

\subsection{FL Model}
The FL problem studied in this paper mostly follows that in the original paper \cite{mcmahan2017fl}. In particular, we consider an FL system with one central parameter server (e.g., base station) and a set of at most $N$ clients (e.g., mobile devices). Client $k \in [N] \triangleq \{1, 2, \cdots, N\}$ stores a local dataset $\mathcal{D}_k = \{\vect{\xi}_i\}_{i=1}^{D_k}$, with its size denoted by $D_k$, that never leaves the client. Datasets across clients are assumed to be non-i.i.d. and disjoint. The maximum data size when all clients participate in FL is $D = \sum_{k=1}^N D_k$. 
Each data sample $\vect{\xi}$ is denoted by an input-output pair $\{\vect{x}, y\}$ for a supervised learning task. We use $f_k(\wbf)$ to denote the local loss function at client $k$, which measures how well an ML model with parameter $\wbf \in \mathbb{R}^d$ fits its local dataset. The global objective function over all $N$ clients is
$
    f(\wbf) = \sum_{k \in [N]} p_k f_k(\wbf),
$
where $p_k = \frac{D_k}{D}$ is the weight of each local loss function, and the purpose of FL is to distributively find the optimal model parameter $\wbf^*$ that minimizes the global loss function: 
$
    \wbf^* \triangleq \argmin_{\wbf\in\mathbb{R}^d}f(\wbf).
$
Let $f^*$ and $f_k^*$ be the minimum value of $f(\wbf)$ and $f_k(\wbf)$, respectively. Then, $\Gamma = f^* - \sum_{k=1}^N \frac{D_{k}}{D}f_k^*$ quantifies the degree of non-i.i.d. as defined in \cite{li2019convergence}.

Specifically, the FL pipeline \cite{mcmahan2017fl} iteratively executes the following steps at the $t$-th learning round.

\begin{enumerate}
\item \textbf{Downlink communication.} The BS broadcasts the current global model $\wbf_t$ to $K$ {randomly selected} clients over the downlink wireless channel. We use $[K]$ to denote the selected client set to simplify the notation, but this should be interpreted as possibly different sets of clients at different round $t$.
\item \textbf{Local computation.} Each selected client uses its local dataset to train a local model improved upon the received global model $\wbf_t$. We assume that mini-batch SGD is used to minimize the local loss function. The parameter is updated iteratively (for $E$ steps) at client $k$ as: $\wbf_{t,0}^k = \wbf_t; \wbf_{t,\tau}^k = \wbf_{t,\tau-1}^k - \eta_t \nabla \tilde{f}_k(\wbf_{t,\tau - 1}^k), \forall \tau = 1, \cdots, E; \wbf_{t+1}^k = \wbf_{t,E}^k$, where $\nabla\tilde{f}_k(\wbf)$ denotes the mini-batch SGD operation at client $k$ on model $\wbf$, and $\eta_t$ is the learning rate (step size).
\item \textbf{Uplink communication.} Each selected client uploads its latest local model to the server synchronously over the uplink wireless channel.
\item \textbf{Server aggregation.} The BS aggregates the received noisy local models $\tilde \wbf_{t+1}^k$ to generate a new global model: $\wbf_{t+1} = \Sigma_{k\in [K]} \tilde p_{k} \tilde \wbf_{t+1}^k$, where $\tilde p_{k} \triangleq \frac{D_k}{\Sigma k \in [K] D_k}$. For simplicity, we assume that each local dataset has equal size, \zixiangb{hence $\tilde p_{k} = {1}/{K}$. }
\end{enumerate}

This work focuses on \textit{both downlink and uplink} communication design in the FL pipeline. We \zixiangb{next} describe the communication models under consideration. 


\subsection{Communication Model}
Consider a MIMO \zixiangb{TDD} communication system equipped with $M$ antennas at the BS (server) where $K$ randomly-selected single-antenna devices (clients) are involved in the $t$-th round of the aforementioned FL task. Let $\hbf_k\in\mathbb{C}^{M\times1}$ denote the \zixiangb{uplink} wireless channel between the $k$-th client and the BS. During the uplink communication phase, each client transmits the difference between the received global model and the newly computed local model 
 \begin{equation}\label{eq:diffModel}
    \mathsf{x}_t^k = \wbf_t - \wbf_{t+1}^k\in\mathbb{R}^d,\;\;\forall k \in [K]
\end{equation}
to the BS, where ${\mathsf{x}_t^k} \triangleq [x_{1,t}^k, \cdots, x_{d,t}^k]^T$. To simplify the notation, we omit index $t$ by using  $x_{k,i}$ instead of $x_{i,t}^k$ barring any confusion. We assume that each client transmits every element of the differential model $\{x_{k,i}\}_{i = 1}^d$ via $d$ shared time slots\footnote{In general, differential model parameters can be transmitted over any $d$ shared orthogonal communication resources (e.g., time or frequency). For simplicity, we use $d$ time slots here.}. For a given element $x_{k,i}$, the received signal at the BS is $\ybf_i^{\ssf{UL}} = \sqrt{P_{\ssf{Client}}}\sum_{k\in[K]} {\bf h}_k x_{k,i} +  \nbf_i, \forall i = 1,\cdots,d$, 
where $P_{\ssf{Client}}$ is the maximum transmit power of each client, and $\nbf_i\in\mathbb{C}^{M\times1}$ represents the uplink noise. Denoting $\Hbf \triangleq \squab{\hbf_1, \cdots, \hbf_K}\in\mathbb{C}^{M\times K}$ and $\xbf_i \triangleq \squab{x_{1,i},\cdots,x_{K,i}}^T \in\mathbb{R}^{K\times 1}, \forall i = 1,\cdots, d$, the received signal\footnote{For simplicity, we assume real signals $\{x_{k,i}\}_{i = 1}^d$ are transmitted in this paper. It can be easily extended to complex signals by stacking two real model parameters into a complex signal, so that the full d.o.f. is utilized.} can be written as
\begin{equation}\label{eq:MatrixForm}
    \ybf_i^{\ssf{UL}} = \sqrt{P_{\ssf{Client}}}\Hbf \xbf_i + \nbf_i.
\end{equation}
It is easy to see that \eqref{eq:MatrixForm} is a standard MIMO communication model and traditional MIMO estimators can be adopted to estimate $\hat\xbf_i = \squab{\hat x_{1,i},\cdots,\hat x_{K,i}}^T$. However, as discussed before, decoding $\{x_{k,i}\}_{i = 1}^d$ individually and obtaining the aggregated parameter $\tilde{x}_i\triangleq\sum_{k\in[K]} \hat x_{k,i}$ by a summation is inefficient. 
After \zixiangb{the} BS decoding all aggregated parameter $\tilde \xbf_t \triangleq \squab{\tilde{x}_1, \cdots,\tilde{x}_d}^T$ in $d$ slots, it can compute the new global model as
\begin{equation}\label{eq:diffGlobal}
    \wbf_{t + 1} = \wbf_t + \frac{1}{K}\tilde \xbf_t.
\end{equation}

\zixiangb{In the downlink, after the computation of the global model $\wbf_{t + 1} = [w_{1, t + 1}, \cdots, w_{d, t + 1}]^T$,} the BS broadcasts the global model to all clients via a precoder $\fbf\in\mathbb{C}^{M \times 1}$, and the received signal at client $k$ is given by
\begin{equation}\label{eq:DLcomp}
    y_i^{\ssf{DL}} = \sqrt{P_{\ssf{BS}}}\hbf^H_{k,t+1}\fbf w_{i, t + 1} + z_i^k,\;\;\forall i = 1,\cdots,d,
\end{equation}
where $P_{\ssf{BS}}$ is the maximum transmit power of the BS and $z_i^k $ denotes the downlink noise. We note that channel $\hbf_{k,t+1}^H\in\mathbb{C}^{1\times M}$ denotes the downlink vector channel that is reciprocal of the uplink channel in round $t+1$. 
Each client \zixiangb{then} computes an estimated global model and uses it as a new initial point for the next learning round after all $d$ elements are received via \eqref{eq:DLcomp}. 
Traditionally, the precoder design of $\fbf$ belongs to broadcasting common messages (see \cite{sidiropoulos2004broadcasting} and the references therein). However, existing methods become impractical due to the difficulty in obtaining full CSI in massive MIMO systems, which motivates us to design $\fbf$ with only partial CSI. For mathematical simplicity, we assume a normalized symbol power\footnote{The parameter normalization and de-normalization procedure in wireless FL follows the same as that in the Appendix of \cite{zhu2019broadband}.}, i.e., $\expt\norm{x_{k,i}}^2 = 1$ and $\expt\norm{w_{i,t + 1}}^2 = 1$; {normalized} Rayleigh block fading channels\footnote{{The large-scale pathloss and shadowing effect is assumed to be taken care of by, e.g., open loop power control \cite{SesiaLTE}.} } $\hbf_k\sim\mathcal{CN}(0,{\frac{1}{M}\bf I})$ in $d$ slots; and i.i.d. Gaussian noise $\nbf_i \sim\mathcal{CN}(0,\frac{\sigma_{\ssf{UL}}^2}{M}{\bf I})$ and $z_i^k \sim\mathcal{CN}(0,\sigma_{\ssf{DL}}^2)$. 
We define the signal-to-noise ratio (SNR) as $\ssf{SNR}_{\ssf{UL}} \triangleq {P_{\ssf{Client}}}/{\sigma_{\ssf{UL}}^2}$ for uplink communications and $\ssf{SNR}_{\ssf{DL}} \triangleq {P_{\ssf{BS}}}/{\sigma_{\ssf{DL}}^2}$ for downlink communications, and without loss of generality (w.l.o.g.) we set $P_{\ssf{Client}} = 1$ and $P_{\ssf{BS}} = 1$. 

\section{Random Orthogonalization}
\label{sec:RO}

In this section, we present the key ideas of random orthogonalization. With this principle, the global model can be directly obtained at the BS via a simple operation in the uplink communications, and the global model can be broadcast to clients efficiently in the downlink communications. By exploring favorable propagation and channel hardening in massive MIMO, our proposed methods only require \emph{partial} CSI, which significantly reduces the channel estimation overhead.


\begin{figure*}
    \centering
    \includegraphics[width = \linewidth]{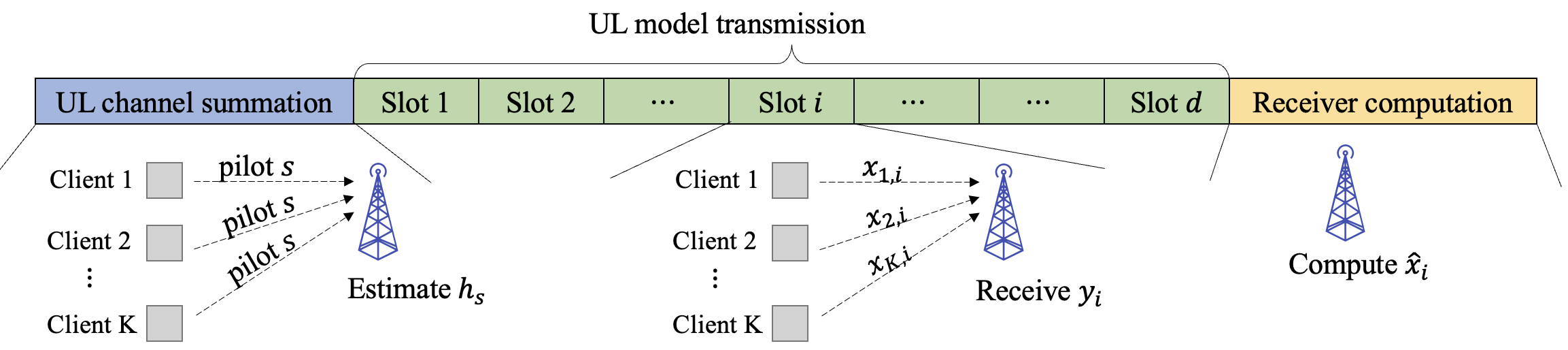}
    \caption{An illustration of the proposed uplink FL design with massive MIMO.}
    \label{fig:framework}
    \vspace{-0.2in}
\end{figure*}

\subsection{Uplink Communication Design}
The designed framework contains the following three main steps in the uplink communications.

\mypara{(U1) Uplink channel summation.}
The BS first schedules all clients participating in the current learning round to transmit a \emph{common} pilot signal $s$ synchronously. The received signal at the BS is
\begin{equation}\label{eq:ULcom}
    \ybf_s = \sum_{k\in[K]} \hbf_k s + \nbf_s,
\end{equation}
and the BS can estimate the {\em summation} of channel {vectors} $\hbf_s \triangleq \sum_{k \in [K]} \hbf_k$ from the received signal $\ybf_s$ (e.g., via a maximum likelihood estimator). We note that the complexity of this sum channel estimation does not scale with $K$. For the purpose of illustrating our key ideas, we assume perfect summation channel estimation at the BS for now.

\mypara{(U2) Uplink model transmission.} All selected clients transmit model differential parameters $\{x_{k,i}\}_{i=1}^d$ to the BS in $d$ shared time slots. The received signal for each differential model element is  $\ybf_{i} = \sum_{k \in [K]} \hbf_k x_{k,i} + \nbf_i, \forall i = 1,\cdots, d$. 

\mypara{(U3) Receiver computation.} The BS estimates each aggregated model element via the following simple \emph{linear projection} operation:
\begin{align}
    & \tilde x_i = \hbf_s^H \ybf_{i} = \sum_{k \in [K]} \hbf_k^H\sum_{k \in [K]} \hbf_k x_{k,i} +  \sum_{k \in [K]} \hbf_k^H\nbf_i \nonumber\\
    & \overset{(a)}{=} \underbrace{\sum_{k \in [K]} \hbf_k^H \hbf_k x_{k,i}}_{\text{Signal}} + \underbrace{\sum_{k \in [K]}\sum_{j \in [K], j\neq k} \hbf_k^H \hbf_j x_{j,i}}_{\text{Interference}}  +  \underbrace{\sum_{k \in [K]} \hbf_k^H\nbf_i}_{\text{ Noise}} \overset{(b)}{\approx} \sum_{k \in [K]} x_{k,i},\;\;\forall i = 1,\cdots, d. \label{eq:OTAComp}
\end{align}
The above three-step uplink communication procedure is illustrated in Fig.~\ref{fig:framework}. Based on Eqn.~\eqref{eq:OTAComp}, the BS then computes the global model via Eqn.~\eqref{eq:diffGlobal} and begins the downlink global model broadcast.

As shown in (a) of Eqn.~\eqref{eq:OTAComp}, inner product $\hbf_s^H \ybf_{i}$ can be viewed as the combination of three parts: signal, interference, and noise. We next show that, taking advantage of two fundamental properties of massive MIMO, the error-free approximation (b) in \eqref{eq:OTAComp} is asymptotically accurate (as the number of BS antennas $M$ goes to infinity). 

\mypara{Channel hardening.} Since each element of $\hbf_k$ is i.i.d. complex Gaussian, by the law of large numbers, massive MIMO enjoys channel hardening \cite{ngo2014aspects}: $\hbf_k^H\hbf_k\rightarrow 1$, as $M\rightarrow\infty$. 
In practical systems, when $M$ is large but finite, for the signal part of (\ref{eq:OTAComp}), we have 
\begin{equation}
    \expt_{\hbf}\squab{\sum_{k \in [K]} \hbf_k^H \hbf_k x_{k,i}} = \sum_{k \in [K]}x_{k,i},
\quad\text{and}\quad
    \variance_{\hbf}\squab{\sum_{k \in [K]} \hbf_k^H \hbf_k x_{k,i}} = \frac{\sum_{k \in [K]}x_{k,i}^2}{M}.
\end{equation}

\mypara{Favorable propagation.} Since channels between different users are independent random vectors, massive MIMO also offers favorable propagation \cite{ngo2014aspects}: $\hbf_k^H\hbf_j\rightarrow 0$, as $M\rightarrow\infty$, $\forall k\neq j$.
Similarly, when $M$ is finite, we have
\begin{equation}
    \expt_{\hbf}\squab{\sum_{k \in [K]}\sum_{j \in [K], j\neq k} \hbf_k^H \hbf_j x_{j,i}} = 0,
\quad\text{and}\quad
    \variance_{\hbf}\squab{\sum_{k \in [K]}\sum_{j \in [K], j\neq k} \hbf_k^H \hbf_j x_{j,i}} = \frac{(K-1)\sum_{k \in [K]}x^2_{k,i}}{M}.
\end{equation}
Furthermore, the expectation of the noise part in \eqref{eq:OTAComp} is zero. Therefore, $\tilde{x}_i$ in \eqref{eq:OTAComp} is an unbiased estimate of the average model. For a given $K$, the variances of both signal and interference decrease in the order of $\mathcal{O}(1/M)$, which shows that \textit{massive MIMO offers {\bf random orthogonality} for analog aggregation over wireless channels}. In particular, the asymptotic element-wise orthogonality of channel vector ensures channel hardening, and the asymptotic vector-wise orthogonality among different wireless channel vectors provides favorable propagation. Both properties render the linear projection operation $\hbf_s^H \ybf_{i}$ an ideal fit for the server model aggregation in FL. 

To gain some insight of random orthogonality, we approximate the average signal-to-interference-plus-noise-ratio (SINR) after the operation in \eqref{eq:OTAComp} as
\begin{equation}\label{eq:SINR} 
\begin{split}
    \expt[\text{SINR}_i] \approx  \frac{\expt_{\hbf,x}\norm{\sum_{k \in [K]} \hbf_k^H \hbf_k x_{k,i}}^2}{\expt_{\hbf,\nbf,x}\norm{\sum_{k \in [K]}\sum_{j \in [K], j\neq k} \hbf_k^H \hbf_j x_{j,i} + \sum_{k \in [K]} \hbf_k^H\nbf_i}^2}
     = \frac{M}{K - 1 + 1/\ssf{SNR}},
\end{split}
\end{equation}
which grows linearly with $M$ for a fixed $K$. On the other hand, for a given number of antennas $M$, Eqn.~\eqref{eq:SINR} can be used to guide the choice of $K$ in each communication round to satisfy an SINR requirement. We will provide more details on the scalability of clients via the convergence analysis of FL with random orthogonalization in Section \ref{subsec:CvgAnaUL}. We note that Eqn.~\eqref{eq:SINR} is an approximate expression for SINR but it sheds light into the relationship between $K$ and $M$. This approximation, however, is not used in the convergence analysis of FL with random orthogonalization in Section \ref{subsec:CvgAnaUL}.
\begin{figure}
    \centering
    \includegraphics[width = .9\linewidth]{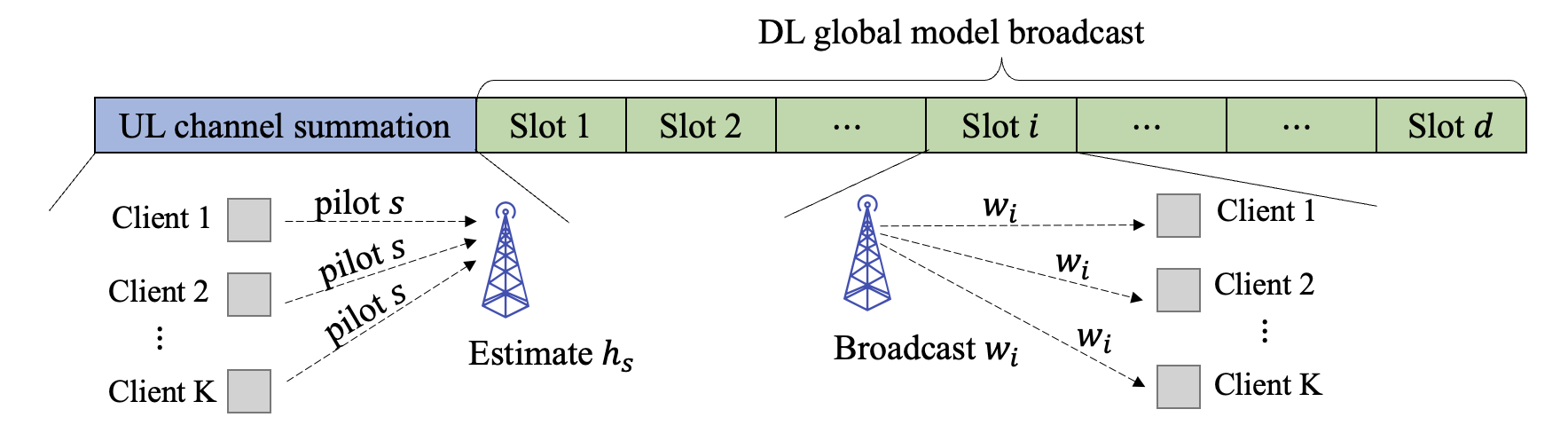}
    \caption{An illustration of the proposed downlink FL design with massive MIMO.}
    \label{fig:frameworkDL}
\end{figure}
\subsection{Downlink Communication Design}
Inspired by the uplink communication design, the downlink design contains the following two steps.

\mypara{(D1) Uplink channel summation.}  This step remains the same as \textbf{U1} in the uplink design. We similarly assume perfect sum channel estimation $\hbf_s = \sum_{k \in [K]} \hbf_k$ at the BS. 

\mypara{(D2) Downlink global model broadcast.} The BS broadcasts global model $\{w_i\}$ to all users, using the estimated summation channel $\hbf_s$ as the precoder. Hence the received signal at the $k$-th user is
\begin{equation}\label{eq:DLOTAComp}
\begin{split}
    y_k = \hbf_k^H \hbf_s w_i + z_i^k 
     \overset{(a)}{=}  \underbrace{\hbf_k^H \hbf_k w_i}_{\text{Signal}} + \underbrace{\sum_{j \in [K], j\neq k} \hbf_k^H \hbf_j w_i}_{\text{Interference}}  + \underbrace{z_i^k}_{\text{Noise}}
     \overset{(b)}{\approx} w_i\;\;\forall i = 1,\cdots, d.
\end{split}
\end{equation}
The above two-step downlink communication procedure is illustrated in Fig.~\ref{fig:frameworkDL}. Similar to the uplink case, the global model signal obtained at each client can also be regarded as the combination of three parts: signal, interference, and noise as shown in \eqref{eq:DLOTAComp}. Leveraging channel hardening and favorable propagation of massive MIMO channels as mentioned before, we have  
\begin{equation}
    \expt_{\hbf}\squab{ \hbf_k^H \hbf_k w_{i}} = w_{i}\quad
\text{and}\quad
    \variance_{\hbf}\squab{ \hbf_k^H \hbf_k w_{i}} = \frac{w_{i}^2}{M},
\end{equation}
for the signal part of (\ref{eq:DLOTAComp}). Besides, we have 
\begin{equation}
    \expt_{\hbf}\squab{\sum_{j \in [K], j\neq k} \hbf_k^H \hbf_j w_{i}} = 0
\quad
\text{and}\quad
    \variance_{\hbf}\squab{\sum_{j \in [K], j\neq k} \hbf_k^H \hbf_j w_{i}} = \frac{(K-1)w^2_{i}}{M},
\end{equation}
for the interference part. The above derivation demonstrates that, similar to the uplink design, received signals obtained via \eqref{eq:DLOTAComp} are unbiased estimates of global model parameters whose variances decrease in the order of $\mathcal{O}(1/M)$ with the increase of BS antennas. We next give a few remarks about the proposed uplink and downlink communication designs of FL with random orthogonalization.

\begin{remark}
In uplink communications, unlike the analog aggregation method in \cite{zhu2019broadband}, the proposed random orthogonalization does not require any individual CSIT. On the contrary, it only requires \zixiangb{partial CSIR, i.e., the estimation of a summation channel $\hbf_s$}, which is $1/K$ of the channel estimation overhead compared with the AirComp method in \cite{yang2020federated} or the traditional MIMO estimators. In downlink communications, the traditional precoder design for common message broadcast requires CSIT for each client. By using the summation channel $\hbf_s$ as the precoder for global model broadcast, only partial CSIT is needed. 
Since we assume a TDD system configuration, the downlink summation channel $\hbf_s$ can be estimated at a low cost utilizing channel reciprocity as shown in Step D1. 
Therefore, the proposed method is attractive in wireless FL due to its mild requirement of {partial CSI}. Moreover, the server obtains global models directly after a series of simple linear projections, which improves the privacy and reduces the system latency as a result of the extremely low computational complexity of random orthogonalization. The same applies to the downlink phase.
\end{remark}

\begin{remark}
Note that although we assume i.i.d. Rayleigh fading channels across different clients, the proposed random orthogonalization method is still valid for other channel models as long as channel hardening and favorable propagation are offered. In massive MIMO millimeter-wave (mmWave) communications, Rayleigh fading channels and light-of-sight (LOS) channels represent two extreme cases: rich scattering and no scattering. It is shown in \cite{ngo2014aspects} that both channel models offer asymptotic channel hardening and favorable propagation. We will discuss the general case that lies in between these two extremes in the next section as well as in the experiment results.  
\end{remark}

\begin{figure}
    \centering
   \includegraphics[width = .93\linewidth]{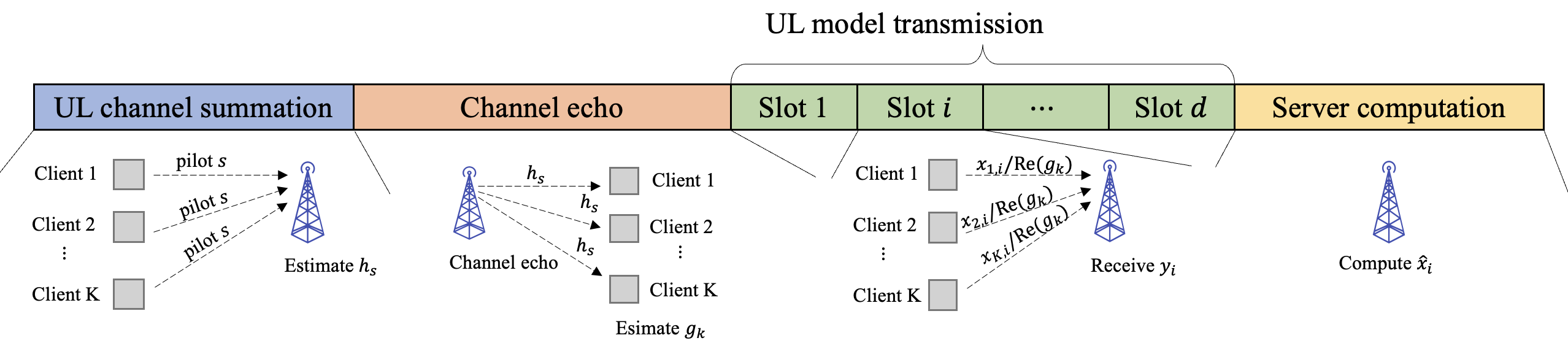}
    \caption{An illustration of the proposed enhanced uplink FL design with massive MIMO.}
    \label{fig:frameworkEnhancedUL}
\end{figure}
\begin{figure}
    \centering
    \includegraphics[width = .93\linewidth]{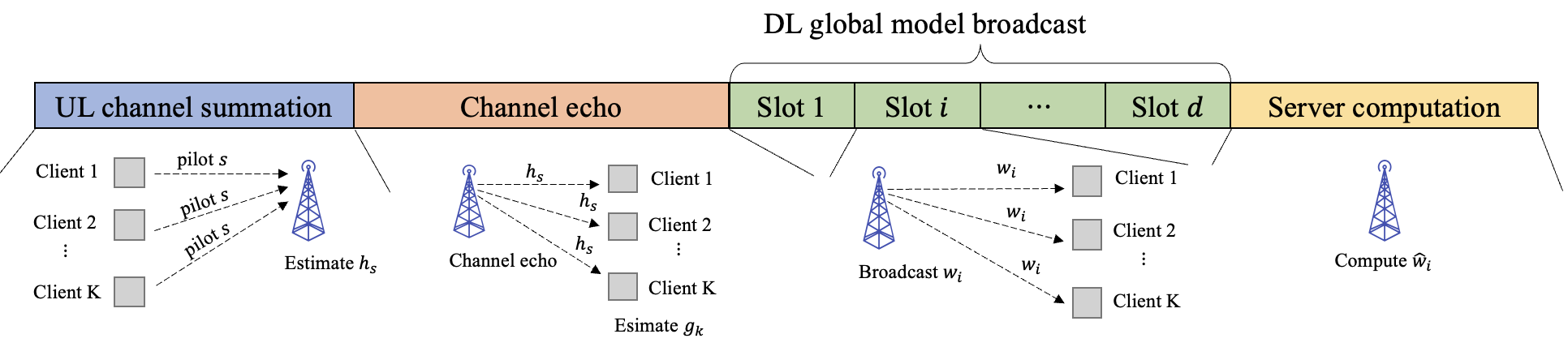}
    \caption{An illustration of the proposed enhanced downlink FL design with massive MIMO.}
    \label{fig:frameworkEnhancedDL}
\end{figure}

\section{Enhanced Random Orthogonalization Design}
\label{sec:enhanced}

The proposed random orthogonalization principle in Section \ref{sec:RO} requires channel hardening and favorable propagation. Although these two properties are quite common in massive MIMO systems as discussed before, in the case that channel hardening and favorable propagation are not available (e.g., the number of BS antennas is small), our design philosophy can still be applied by introducing a novel \textbf{channel echo mechanism}. In this section, we present an enhanced design to the methods in Section \ref{sec:RO} by taking advantage of channel echos.

Channel echo refers to that the receiver sends whatever it receives back to the original transmitter as the data payload. The main purpose of channel echo in the uplink communication of FL is to ``cancel" channel fading at each client. The enhanced design for the uplink communications contains the following four main steps, which is demonstrated in Fig.~\ref{fig:frameworkEnhancedUL}.


\mypara{(EU1) Uplink channel summation.} The first step of the enhanced design follows the same as the random orthogonalization method (U1 and D1), so that the BS has the estimate sum channel vector $\hbf_s = \sum_{j\in [K]} \hbf_k$.

\mypara{(EU2) Downlink channel echo.} The BS sends the previously estimated $\hbf_s$ (after normalization to satisfy the power constraint) to all clients. For the $k$-th client, the received signal is $\ybf_k = \frac{\hbf_k^H\hbf_s}{\sqrt{K}} + \nbf_k$,
by which client $k$ can estimate $g_k = \hbf_k^H\hbf_s = \hbf_k^H\sum_{j\in [K]} \hbf_k = \norm{\hbf_k}^2 + \sum_{j \in [K], j\neq k}^K \hbf_k^H \hbf_j$. Note again that we assume a perfect estimation of $\hbf_s$. An additional error term will appear in the estimation of $g_k$ when the summation channel estimation is imperfect, which will be discussed later.

\mypara{(EU3) Uplink model transmission.} All involved clients transmit local parameter $\{{x_{k,i}}/{\real(g_k)}\}_{k\in[K]}$ to the BS synchronously in $d$ shared time slots: 
$
    \ybf_{i} = \sum_{k\in [K]}^K \hbf_k \frac{x_{k,i}}{\real(g_k)} + \nbf_i,\;\;\forall i = 1,\cdots,d.$

\mypara{(EU4) Server computation.} The BS obtains $\sum_{k\in[K]} x_{k,i}$ via the following operation:
\begin{equation}\label{eq:ehancedUL}
\begin{split}
    \tilde x_i &= \real(\ybf_{i}^H\hbf_s) = \real\squab{\sum_{k \in [K]} \hbf_k^H \frac{x_{k,i}}{\real(g_k)} \sum_{j\in[K]}^K \hbf_j+ \nbf_i^H \sum_{j\in[K]} \hbf_j} \\
    & 
    = \sum_{k\in[K]} \frac{x_{k,i}}{\real(g_k)}\real\squab{\hbf_k^H  \sum_{j \in [K]} \hbf_j} + \real\squab{\nbf_i^H \sum_{j\in [K]} \hbf_j} = \sum_{k\in [K]} x_{k,i} + \real\squab{ \sum_{j\in [K]} \hbf_j^H \nbf_i}.
\end{split}
\end{equation}

Similarly, as shown in Fig.~\ref{fig:frameworkEnhancedDL}, 
the enhanced design for the downlink communication contains the following four main steps.

\mypara{(ED1-2) Uplink channel summation and downlink channel echo.} The first two steps in the downlink design remain the same as Steps EU1 and EU2 in the uplink design, so that the BS can estimate channel vector summation $\hbf_s = \sum_{j\in [K]} \hbf_k$ and each client can estimate the parameter $g_k$.

\mypara{(ED3) Downlink global model broadcast.} The BS broadcasts global model $\{w_i\}$ to all clients using the estimated sum channel $\f{\hbf_s}{\sqrt{K}}$ as the precoder. The received signal at the $k$-th client is $ y_k = \hbf_k^H \f{\hbf_s}{\sqrt{K}} w_i + \nbf_i = \f{1}{\sqrt{K}} g_k w_i + z_i^k, \forall i = 1,\cdots, d$.

\mypara{(ED4) Model parameter computation.} Each user obtains the global model $\{w_i\}$ via the following calculation:
\begin{equation}\label{eq:ehancedDL}
    \real\squab{\frac{\sqrt{K}y_k}{g_k}} =  w_i + \real(\f{\sqrt{K}z_i^k}{g_k}) \approx \hat w_i, \;\; \forall i = 1,\cdots, d.
\end{equation}

Compared with the random orthogonalization method that offers \emph{asymptotic} interference-free global model estimation, the received FL parameters obtained by the enhanced method are \textit{completely interference-free} at both the server and the clients, as shown in \eqref{eq:ehancedUL} and \eqref{eq:ehancedDL}. The extra channel echo steps (Step EU2 in uplink and Step ED1 in downlink) allow clients to obtain \emph{partial} CSI $g_k$, so that they can pre-cancel and post-cancel channel interference among different user channels in uplink and downlink communications, respectively. Therefore, \textbf{this enhancement is valid even if channel hardening and favorable  propagation are not present in wireless channels}, at a low cost of using one extra slot for the channel echo operation, and preserves all the other advantages of random orthogonalization. 

\begin{remark}
We note that both random orthogonalization and enhanced methods assume a perfect estimation of $\hbf_s$. In practical systems, to improve the accuracy of the estimate $\hat{\hbf}_s$, BS can use multiple pilots / time slots for channel estimation. Moreover, for the enhanced method, the estimation error of $\hbf_s$ itself will not affect the performance, since the imperfect estimated summation channel will cancel out in Step EU/ED$4$. Only the imperfect estimation of $g_k$ will influence the results. We provide more details on the robustness of the proposed schemes over imperfect $\hat{\hbf}_s$ and $g_k$ in the experiment results.
\end{remark}

\begin{remark}
In the enhanced uplink design, each client pre-cancels the channel fading effect so that the global model can be directly obtained at the BS after simple operations. Note that the analog aggregation method in \cite{zhu2019broadband} also uses ``channel inversion" to pre-cancel channel fading. However, our design outperforms the method in \cite{zhu2019broadband} because the latter requires full CSIT, which leads to a large channel estimation overhead even with channel reciprocity in a TDD system. On the contrary, our method only requires partial CSI, which can be efficiently obtained via channel echos. Moreover, analog aggregation does not naturally extend to MIMO systems when the uplink channels become vectors, which makes channel inversions at the transmitters nontrivial.
\end{remark}

\section{Performance Analyses}
\label{sec:analysis}

We analyze the performance of the proposed methods from two aspects. On the communication performance side, we derive CRLBs of the estimates of model parameters in both uplink and downlink phases as the theoretical benchmarks. On the machine learning side, we present the convergence analysis of FL when the proposed communication designs are applied.

\subsection{Cramer-Rao Lower Bounds}\label{subsec:crlb}

In the uplink communication, recall that the received signal is $\ybf_{i} = \Hbf\xbf_i + \nbf_i$. Denoting $\mubf_{\ssf{UL}} = \Hbf\xbf_i$, we have $\ybf_{i} \sim \mathcal{CN}(\mubf_{\ssf{UL}}, \frac{1}{\ssf{SNR}}\Ibf)$. To leverage CRLBs to evaluate the parameter estimation, we first need to derive the Fisher information of $\xbf_i$. Based on Example $3.9$ in \cite{kay1993fundamentals}, we can write the Fisher information matrix (FIM) of the estimation of $ \xbf_i$ as:  
\begin{equation*}
    \Fbf_{\ssf{UL}} = 2\cdot  \ssf{SNR} \cdot  \real\squab{\frac{\partial^H\mubf_{\ssf{UL}}(\xbf_i)}{\partial \xbf_i}\frac{\partial\mubf_{\ssf{UL}}(\xbf_i)}{\partial \xbf_i}}.
\end{equation*}
After inserting $\frac{\partial\mubf_{\ssf{UL}}(\xbf_i)}{\partial \xbf_i} = \Hbf$ into the FIM, we have $\Fbf_{\ssf{UL}} = 2 \cdot \ssf{SNR} \cdot \real(\Hbf^H\Hbf)$. Note that for the enhanced uplink design, we can absorb $\real(g_k)$ into the effective channel as  $\tilde\Hbf\triangleq\squab{\hbf_1/\real(g_1),\cdots,\hbf_K/\real(g_K)}$, and calculate FIM via $\Fbf_{\ssf{UL}} = 2 \cdot \ssf{SNR} \cdot \real(\tilde\Hbf^H\tilde\Hbf)$. 

In the downlink communication, since $y_k = \hbf_k^H \hbf_s w_i + \nbf_k$, by the definition of $\mu_{\ssf{DL}} = \hbf_k^H\hbf_s w_i$, we have that $y_{k} \sim \mathcal{CN}(\mu_{\ssf{DL}}, \frac{1}{\ssf{SNR}})$. The Fisher information of global model parameters is
\begin{equation*}
    F_{\ssf{DL}} =  2\cdot  \ssf{SNR} \cdot  \real\squab{\frac{\partial^H\mu_{\ssf{DL}}(w_i)}{\partial w_i}\frac{\partial\mu_{\ssf{DL}}(w_i)}{\partial w_i}} = 2\cdot  \ssf{SNR} \cdot  \real (\hbf_k^H\hbf_s\hbf_s^H\hbf_k).
\end{equation*}

The CRLBs of estimates are then given by the inverse of the Fisher information (matrix): 
    $\Cbf_{\hat\xbf_i} = \Fbf_{\ssf{UL}}^{-1}$ and $C_{\hat w_i} = 1/F_{\ssf{DL}}$, respectively.
CRLBs are the lower bounds on the variances of unbiased estimators, stating that the variance of any such estimator is at least as high as the inverse of the Fisher information (matrix). We have shown that the proposed methods lead to unbiased estimations of the global model in both uplink and downlink communications.  Hence, we can use the sum of all diagonal elements of $\Cbf_{\hat\xbf}$ as the lower bound of the mean squared error (MSE) $\expt\norm{\xbf_i - \hat\xbf_i}^2$, and use $C_{\hat w_i}$ as the lower bound of MSE $\expt\norm{w_i - \hat w_i}^2$, to evaluate the performance of model estimation in both uplink and downlink communications.  These bounds will be validated in the experiment results.

\subsection{ML Model Convergence Analysis}
\label{subsec:CvgAnaUL}

We now analyze the ML model convergence performances of the proposed methods. Note that as we have proposed two different designs (basic and enhanced) for the uplink and downlink communications, respectively, there would be four cases of convergence analysis. Since these convergence analyses are quite similar, we only report one of these results. 

We first make the following standard assumptions that are commonly adopted in the convergence analysis of \textsc{FedAvg} and its variants \cite{li2019convergence,jiang2018nips,stich2018local,Zheng2020jsac}. In particular, Assumption \ref{as:1} indicates that the gradient of $f_k$ is Lipschitz continuous. The strongly convex loss function in Assumption \ref{as:2} is a category of loss functions that are widely studied in the literature (see \cite{li2019convergence} and its follow-up works). Assumptions \ref{as:3} and \ref{as:4} imply that the mini-batch stochastic gradient and its variance are bounded \cite{stich2018local}.

\begin{assumption}\label{as:1}
\textbf{$L$-smooth:} $\forall~\vect{v}$ and $\vect{w}$, $\norm{f_k(\vbf)-f_k(\wbf)}\leq L \norm{\vbf-\wbf}$;
\end{assumption}
\begin{assumption}\label{as:2}
\textbf{$\mu$-strongly convex:} $\forall~\vect{v}$ and $\vect{w}$, $\left<f_k(\vbf)-f_k(\wbf), \vbf-\wbf\right>\geq \mu \norm{\vbf-\wbf}^2$;
\begin{assumption}\label{as:3}
\textbf{Bounded variance for unbiased mini-batch SGD:} $\forall k \in [N]$, $\expt[\nabla\tilde{f}_k(\wbf)] = \nabla f_k(\wbf)\text{~~and~~}\\\expt\norm{\nabla f_k(\wbf) - \nabla \tilde f_k(\wbf)}^2\leq H_k^2$;
\end{assumption}
\begin{assumption}\label{as:4}
\textbf{Uniformly bounded gradient:} $\forall k \in [N]$, $\expt\norm{\nabla\tilde{f}_k(\wbf)}^2 \leq H^2$ for all mini-batch data.
\end{assumption}
\end{assumption}

We next provide a convergence analysis of FL when the uplink communication utilizes random orthogonalization and the enhanced design is applied to the downlink communication. Note that unlike uplink communications, we cannot use model differential for downlink FL communications because of partial clients selection. To guarantee the convergence of FL, we need to borrow the necessary condition for noisy FL downlink communication from our previous work \cite{wei2021TCCN}, i.e., downlink transmit power should scale in the order of $\mathcal{O}(t^2)$.

\begin{theorem}[\textbf{\em Convergence for random orthogonalization in the uplink and enhanced method in the downlink}]
\label{thm.RO} 
Consider a wireless FL task that applies random orthogonalization for the uplink communications and the enhanced method for the downlink communications. With Assumptions 1-4, for some $\gamma\geq 0$, if we set the learning rate as $\eta_t = \frac{2}{\mu(t+\gamma)}$ and downlink SNR scales as $\ssf{SNR}_{\ssf{DL}}\geq \f{1-\mu\eta_t}{\eta^2_t}$ in round $t$, we have
\begin{equation}\label{eq.convergence}
    \expt[f(\wbf_t)] - f^* \leq \frac{L}{2(t + \gamma)}\left[\frac{4B}{\mu^2} + (1 + \gamma)\norm{\wbf_0 - \wbf^*}^2\right],
\end{equation}
for any $t\geq 1$, where 
\begin{equation}\label{eq.Bdef}\small
\begin{split}
      B \triangleq \sum_{k=1}^N \frac{H_k^2}{N^2} + 6L \Gamma + 8(E-1)^2 H^2  +  \frac{N-K}{N-1} \frac{4}{K}  E^2 H^2 + \frac{4}{K}\left(\frac{K}{M} + \frac{1}{\ssf{SNR}_{\ssf{UL}}}\right)  E^2 H^2 + \frac{MK}{N^2(K + M)}.
    \end{split}
\end{equation}
\end{theorem}
\begin{proof}
Proof of Theorem \ref{thm.RO} is given in Appendix \ref{sec:proofThrm1}.
\end{proof}

Theorem \ref{thm.RO} shows that applying random orthogonalization in the uplink communications and enhanced method in the downlink communications preserves the $\mathcal{O}(1/T)$ convergence rate of vanilla SGD in FL tasks with perfect communications in both uplink and downlink phases. The factors that impact the convergence rate are captured entirely in the constant $B$, which come from multiple sources as explained below: $\frac{\sum_{k \in [N]} H_k^2}{N^2}$ comes from the variances of stochastic gradients; $6L \Gamma$ is introduced by the non-i.i.d. of local datasets; the choice of local computation steps and the fraction of partial client participation lead to $8(E-1)^2 H^2$ and $\frac{N-K}{N-1} \frac{4}{K} E^2 H^2 $, respectively; and the interference and noise in uplink and downlink communications result in $\frac{4}{K}\left(\frac{K}{M} + \frac{1}{\ssf{SNR}_{\ssf{UL}}}\right)  E^2 H^2$ and $\left(d + \frac{dK}{M}\right)$, respectively. Note that the impact of the downlink noise, i.e., $\ssf{SNR}_{\ssf{DL}}$, is not explicit in $B$ due to the requirement of $\ssf{SNR}_{\ssf{DL}}\geq \f{1-\mu\eta_t}{\eta^2_t}$ to guarantee the convergence. 

\begin{remark}\label{remark:constantconvergence}
We note that Theorem \ref{thm.RO} considers random orthogonalization in the uplink and enhanced method in the downlink. When random orthogonalization is adopted in the downlink, the convergence bound in \eqref{eq.convergence} will suffer from an additive constant term. This is because the interference cannot be effectively reduced when downlink power scales in the order of $\mathcal{O}(t^2)$, as required for direct model transmission \cite{wei2021TCCN}. This gap is also empirically observed in the experiments (see Section \ref{subsec:learningperformance}). However, we also note that this gap is inversely proportional to the number of antennas $M$. Hence, as $M$ becomes large, it reduces to zero asymptotically\footnote{Due to the space limitation, the technical details for this remark are deferred to our technical report \cite{shenonline1}.}. 
\end{remark}

We next analyze the relationship between the number of selected clients $K$ and the number of BS antennas $M$ to understand the scalability of multi-user MIMO for FL, which provides more insight for practical system design. To this end, we consider a simplified case where the system only configures random orthogonalization in the uplink communications, assuming that the downlink communications are error-free. Note that this configuration is reasonable when the BS has large transmit power. 
We further assume full client participation ($N = K$), one-step SGD at each device ($E = 1$), and i.i.d datasets across all clients ($\Gamma = 0$). For this special case, we establish Corollary \ref{lemma.1} as follows.

\begin{coro}[\textbf{\em Convergence for the simplified case}]\label{lemma.1}
Consider a MIMO system that applies random orthogonalization for the uplink communications of FL with full client participation, one-step SGD at each device, and i.i.d datasets across all clients. Based on Assumptions 1-4 and choosing learning rate as $\eta_t = \frac{2}{\mu(t+\gamma)}$, $\forall t\in [T]$, the following inequality holds:
\begin{equation}\label{eq.convergenceSimplify}
    \expt[f(\wbf_t)] - f^* \leq \frac{L}{2(t + \gamma)}\left[\frac{4\tilde{B}}{\mu^2} + (1 + \gamma)\norm{\wbf_0 - \wbf^*}^2\right]
\end{equation}
for any $t\geq 1$, where 
\begin{equation}\label{eq.BdefSimplify}
\tilde{B}\triangleq \squab{1 + \frac{K}{M} +\f{1}{{\ssf{SNR}} }}\frac{H^2}{K}.
\end{equation}
\end{coro}
\begin{proof}
Corollary \ref{lemma.1} comes naturally from Theorem \ref{thm.RO} by setting $N = K$, $\Gamma = 0$, $E = 1$, omitting the $\left(d + \frac{dK}{M}\right)$ term due to the perfect downlink communications, and the fact that $\expt\norm{\nabla f_k(\wbf) - \nabla \tilde f_k(\wbf)}^2\leq \expt\norm{\nabla \tilde f_k(\wbf)}^2\leq H^2$.
\end{proof}

Corollary \ref{lemma.1} shows that there are two main factors that impact the convergence rate of FL with MIMO: \textbf{variance reduction} and \textbf{channel interference and noise}. In particular, the definition of $\tilde{B}$ in \eqref{eq.BdefSimplify}, which appears in Corollary \ref{lemma.1}, captures the joint impact of both factors.  The nature of distributed SGD suggests that, for a fixed mini-batch size at each client, involving $K$ devices enjoys a $\frac{1}{K}$ variance reduction of stochastic gradient at each SGD iteration \cite{johnson2013accelerating}, which is captured by the $\frac{H^2}{K}$ term in \eqref{eq.BdefSimplify}. However, due to the existence of interference and noise, the convergence rate is determined by both factors, shown as $\frac{H^2}{K}$ and $\frac{(K/M + 1/\ssf{SNR})H^2}{K}\approx \frac{H^2}{M}$ terms in \eqref{eq.BdefSimplify}. This suggests that the desired variance reduction may be adversely impacted if channel interference and noise dominate the convergence bound. In particular, when $M \gg K$, we have $\frac{1}{K} \gg \frac{1}{M}$, and the system enjoys almost the same variance reduction as the interference-free and noise-free case. However, in the case of $K \gg M$, we have $\frac{(K/M + 1/\ssf{SNR})}{K}\approx \frac{1}{M}\gg \frac{1}{K}$, and $\frac{H^2}{M}$ dominates the convergence bound. In this case, it is unwise to blindly increase the number of clients, as it does not have the advantage of variance reduction. 

\begin{remark}
In massive MIMO, a BS is usually equipped with many (up to hundreds) antennas. Although there may be large number of users participating in FL, only a small number of them are simultaneously active \cite{yang2020federated}. Both factors indicate that $K \ll M$ often holds in typical massive MIMO systems. The  analysis reveals that our proposed framework enjoys nearly the same interference-free and noise-free convergence rate with low communication and computation overhead in massive MIMO systems.
\end{remark}

\section{Experiment Results}
\label{sec:sim}
We evaluate the performances of random orthogonalization and the enhanced method for uplink and downlink FL communications through numerical experiments. From a communication performance perspective, we compare the proposed methods with the classic MIMO \zixiangb{estimators and precoders} with respect to the MSE. We provide the computation time comparison as a measure of the complexity of various methods. We also discuss the robustness of the proposed methods when the properties of channel hardening and favorable propagation are not fully offered and the channel estimation is imperfect. We further verify the effectiveness of the proposed methods via FL tasks using real-world datasets. 


\begin{figure}[htb]
    \centering
    \subfigure[]{\includegraphics[width = 0.48\linewidth]{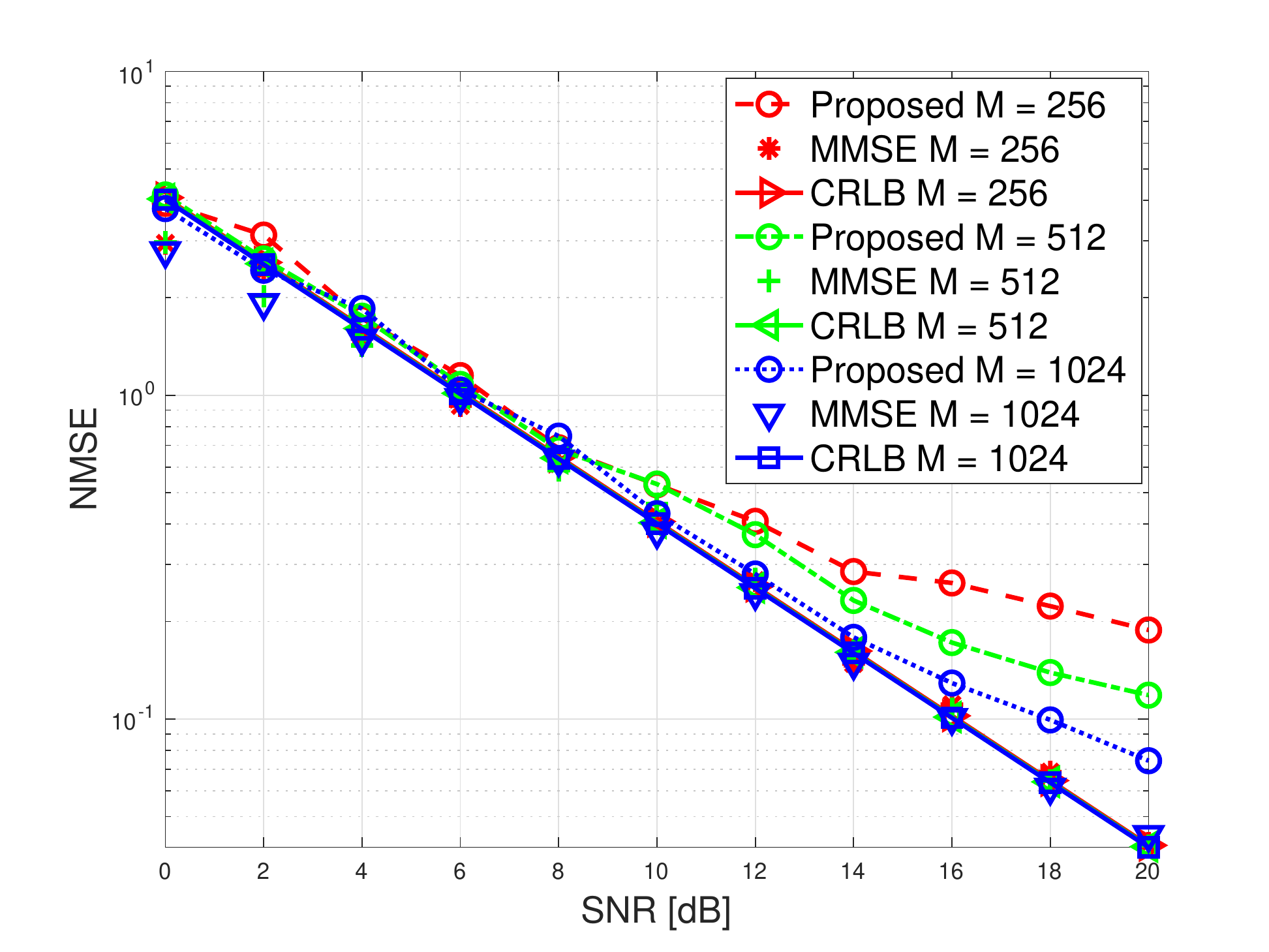}}
    \subfigure[]{\includegraphics[width = 0.48\linewidth]{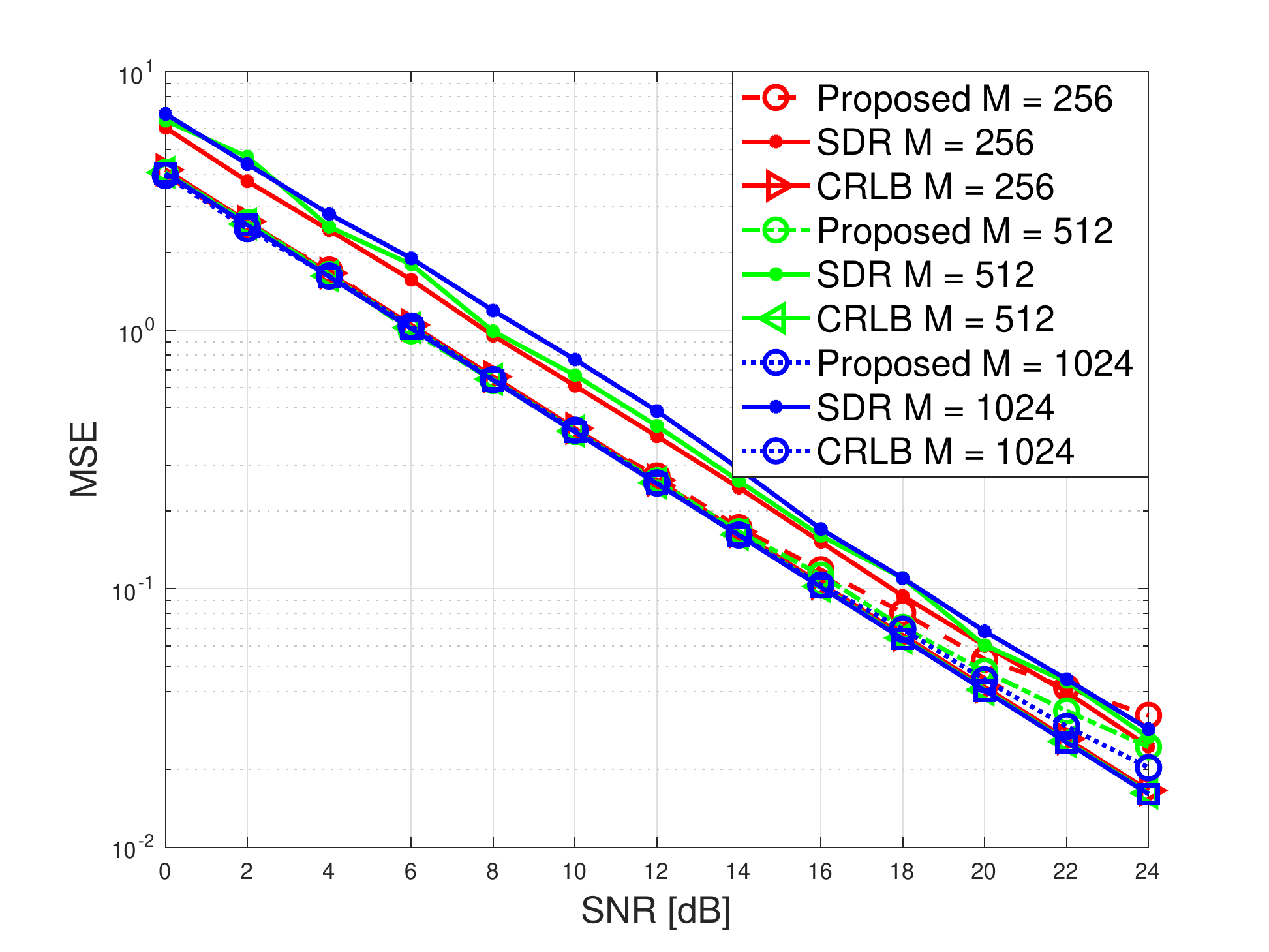}}
    \subfigure[]{\includegraphics[width = 0.48\linewidth]{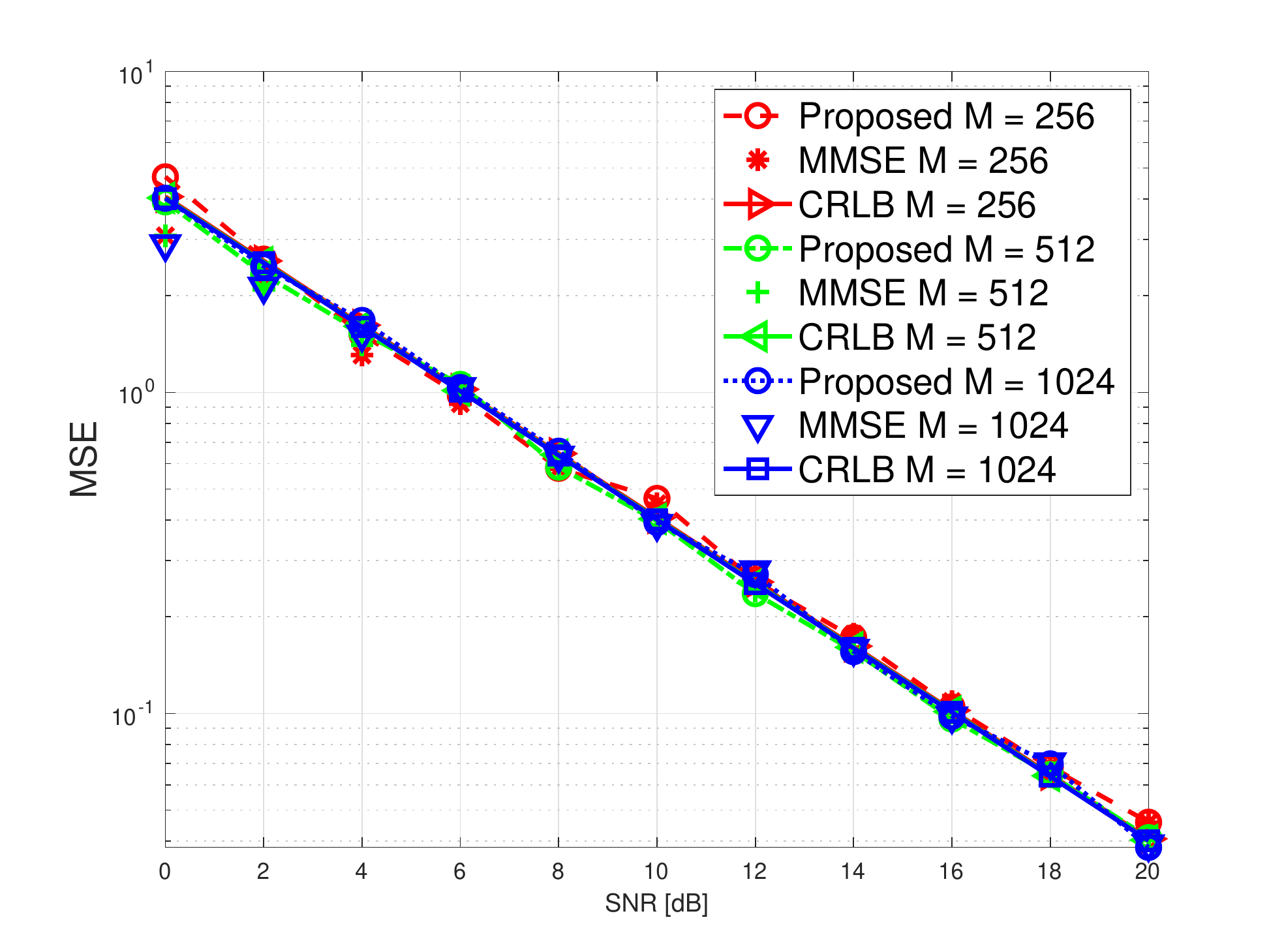}}
    \subfigure[]{\includegraphics[width = 0.48\linewidth]{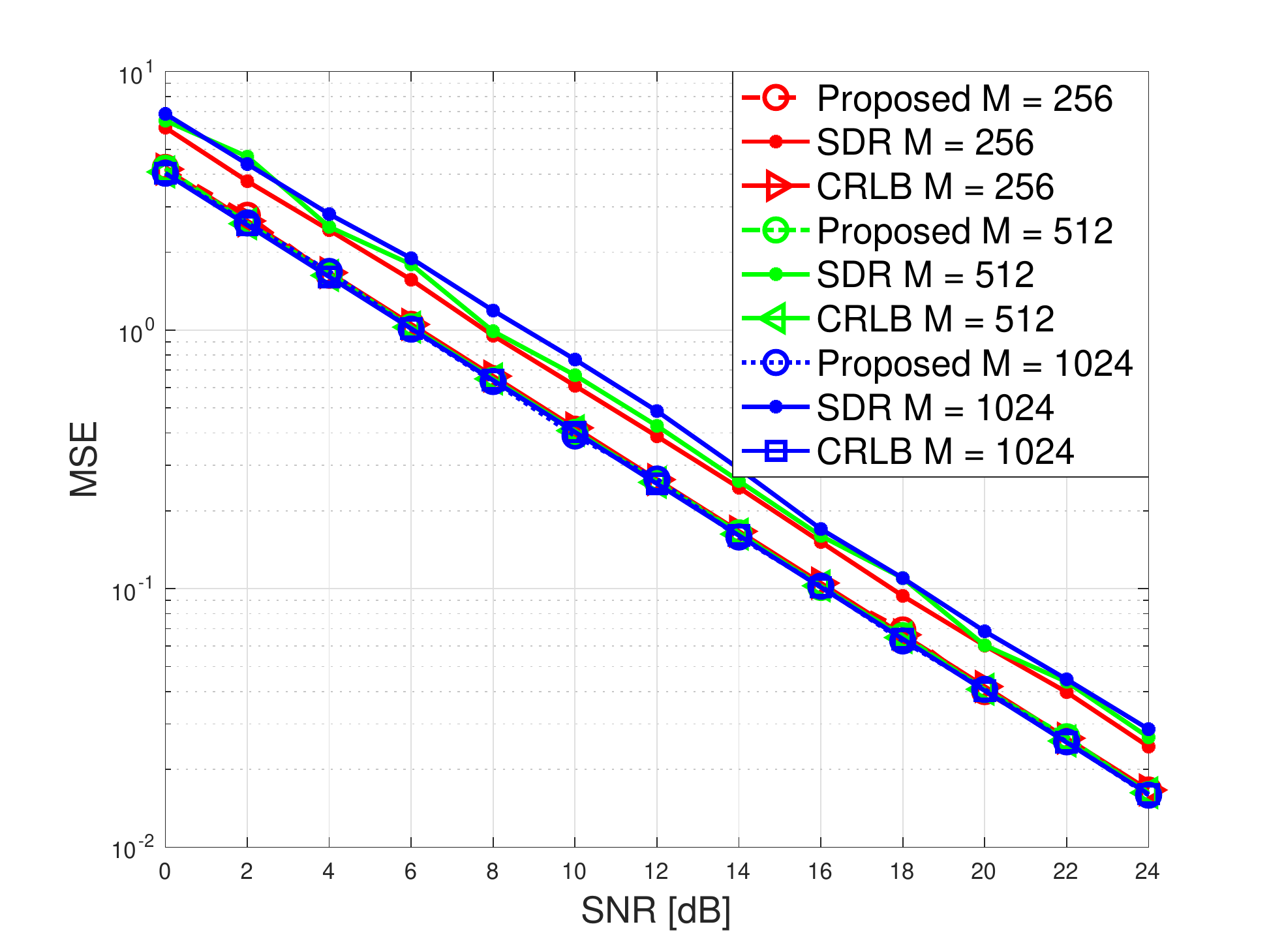}}
    \caption{MSE of the received global ML model parameters versus SNR of random orthogonalization in uplink (a) and downlink (b) communications and of the enhanced method in uplink (c) and downlink (d) communications.}
\label{fig:MSE-RO-enhanced}
\end{figure}

\subsection{Communication Performance}

We consider a massive MIMO BS with $M = 256$, $512$, or $1024$ antennas, with $K = 8$ active users participating in an FL task. We assume a Rayleigh fading channel model, i.e., $\hbf_k \sim \mathcal{CN}(0,\frac{1}{M}\Ibf)$, for each user, and use the MSE of the computed global model parameters in uplink and downlink communications to evaluate the system performance. All MSE results are obtained from $2000$ Monte Carlo experiments. We use CRLBs derived in Section \ref{subsec:crlb} as the benchmarks of the computed MSEs. In addition, we adopt the traditional MIMO MMSE estimator and the semidefinite relaxation based (SDR-based) \zixiangb{precoder} design method in \cite{sidiropoulos2004broadcasting} for performance comparisons of uplink and downlink communications, respectively.

\mypara{Effectiveness.} Fig.~\ref{fig:MSE-RO-enhanced}(a) and Fig.~\ref{fig:MSE-RO-enhanced}(b) compare the MSE performance of the random orthogonalization method in uplink and downlink communications with the traditional MIMO estimator/\zixiangb{precoder} as well as the CRLB under different system SNRs. As illustrated in these two plots, the proposed method performs nearly identically {to} the CRLB in low and moderate SNRs under different antenna configurations (see $\ssf{SNR} \leq 12$ dB for uplink and $\ssf{SNR} \leq 18$ dB for downlink). As the SNR increases, the dominant factor affecting system performances {becomes} the interference among different users. In the uplink communications, when $K\leq M$ and at high SNR, Eqn.~\eqref{eq:SINR} shows that for a given $K$ and $M$, the proposed method has a fixed (approximate) $\ssf{SIR} = \frac{K - 1}{M}$ as $\ssf{SNR} \rightarrow \infty$, which explains why the performance of the proposed scheme deteriorates compared with MMSE at high SNR. However, this issue disappears naturally as the number of BS antennas increases. It can be seen in Fig.~\ref{fig:MSE-RO-enhanced}(a) that the performance gap between the proposed method and the CRLB reduces, from about $7$ dB when $M = 256$ to about $2$ dB when $M = 1024$ at $\ssf{SNR} = 20$ dB, in uplink communications. We note that, although random orthogonalization produces higher MSEs than the MMSE estimator, the FL tasks have the same convergence rate under a constant SINR in uplink communications as indicated by the convergence analysis. This will be further validated in Section \ref{subsec:learningperformance} by showing that random orthogonalization hardly slows the convergence of FL. Similar to uplink, random orthogonalization performs nearly identically {to} the CRLB in low and moderate SNRs in downlink, and only loses about $0.5 \sim 3$ dB under different antenna configurations at $\ssf{SNR} = 24$ dB. Moreover, random orthogonalization outperforms the SDR-based method at almost all SNRs and antenna configurations. Due to its sub-optimality, the SDR-based method has about $1.5$ dB loss compared with the CRLB. We should emphasize that our method only requires $1/K$ of the channel estimation overheard (partial CSI) compared with both MMSE and the SDR-based method (full CSI), and this advantage is more pronounced when the BS is equipped with a larger number of antennas.

Similarly, Fig.~\ref{fig:MSE-RO-enhanced}(c) and Fig.~\ref{fig:MSE-RO-enhanced}(d) compare the MSE performance of the enhanced method in uplink and downlink communications with the MMSE estimator / SDR-based \zixiangb{precoder}. It is clear from both plots that the enhanced method achieves MSEs very closed to the CRLBs. Furthermore, it performs nearly identically as the MMSE estimator in uplink and outperforms the SDR-based method by about $1.5$ dB in downlink. Therefore, by introducing channel echos, the enhanced method achieves excellent performance while consuming relatively low additional resource.

\begin{table}
\caption{Computation time comparison between the proposed methods and the MMSE/SDR Method}
\label{table:comp}\small
    \centering
    \begin{tabular}{c| c c c | c c c}
    \hline
    \# antennas
    &   \multicolumn{2}{c}{Total CPU time (second)} & Ratio & \multicolumn{2}{c}{Total CPU time (second)} & Ratio \cr
    (M) & RO-UL & MMSE & RO-UL/MMSE & Enhanced-UL & MMSE & Enhanced-UL/MMSE\cr \hline
    256 & 0.0186 & 2.7141 & 0.68\% & 0.0203 & 2.9228 & 0.69\% \cr
    512 & 0.0303 & 12.4155 & 0.24\% & 0.0469 & 16.3938 & 0.30\% \cr
    1024 & 0.0448 & 82.3530 & 0.05\% & 0.0711 & 91.4117 & 0.07\% \cr\hline\hline
    \# antennas
    &   \multicolumn{2}{c}{Total CPU time (second)} & Ratio & \multicolumn{2}{c}{Total CPU time (second)} & Ratio \cr
    (M) & RO-DL & SDR & RO-DL/SDR & Enhanced-DL & SDR & Enhanced-DL/SDR\cr 
    \hline
    256 & 0.0157 & 25.1492 & 0.062\% & 0.0163 & 28.8593 & 0.68\% \cr
    512 & 0.0415 & 324.7349 & 0.012\% & 0.0592 & 492.9539 & 0.012\% \cr
    1024 & 0.0571 & 4819.6221 & 0.0012\% & 0.0695 & 5925.9250 & 0.0011\% \cr\hline
    \end{tabular}
\end{table}

\mypara{Efficiency.} We next focus on the low-latency benefit of the proposed methods. Table \ref{table:comp} compares the computational time of the proposed schemes with the MMSE estimator and the SDR-based \zixiangb{precoder} when $\ssf{SNR} = 10$ dB in the uplink and downlink communications, respectively. The total CPU time is the \emph{cumulative time} {of each algorithm} over $2000$ Monte Carlo experiments. We see that the time consumption of random orthogonalization and the enhanced method is much less than that of the MMSE estimator and the SDR-based \zixiangb{precoder}. Especially, when $M = 1024$, despite the $0.3$ dB normalized MSE (NMSE) performance loss of random orthogonalization compared with the MMSE estimator in the uplink communications (as shown in Fig.~\ref{fig:MSE-RO-enhanced}(a)), the computation time of the former is only $0.05\%$ of the latter. Since SDR in general has an $O(M^3)$ complexity, the proposed methods are even more computationally efficient for the downlink communications, as the total CPU time is less than $0.1\%$ of the SDR-based method in all settings.  All these results suggest that both random orthogonalization and its enhancement are attractive in massive MIMO systems, because they have nearly identical MSE performances to CRLBs but require much less channel estimation overhead and achieve extremely lower system latency than the classic MIMO estimators and \zixiangb{precoders}.

\begin{figure}[htb]
    \centering
    \subfigure{\includegraphics[width = 0.45\linewidth]{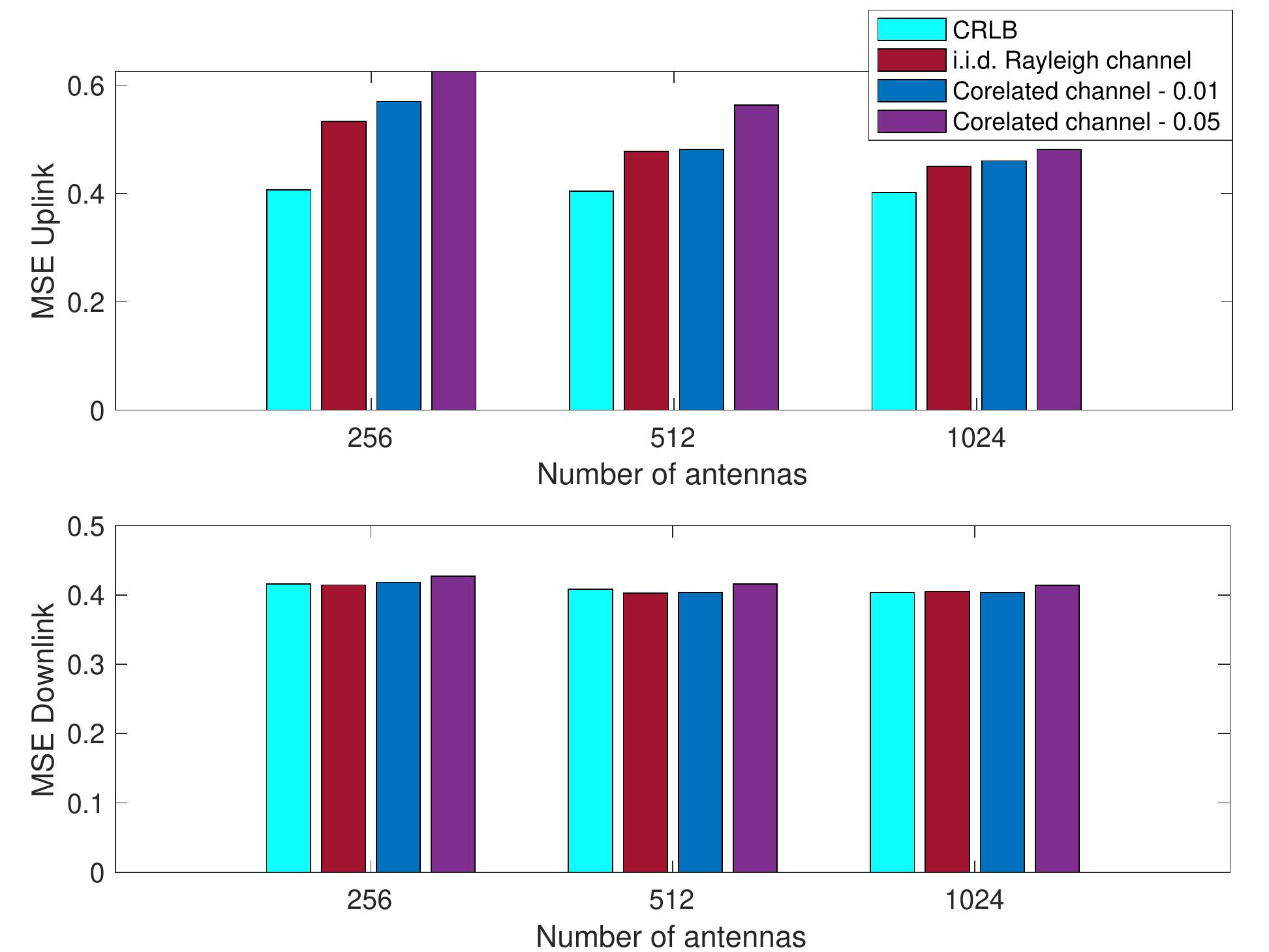}}
    \subfigure{\includegraphics[width = 0.45\linewidth]{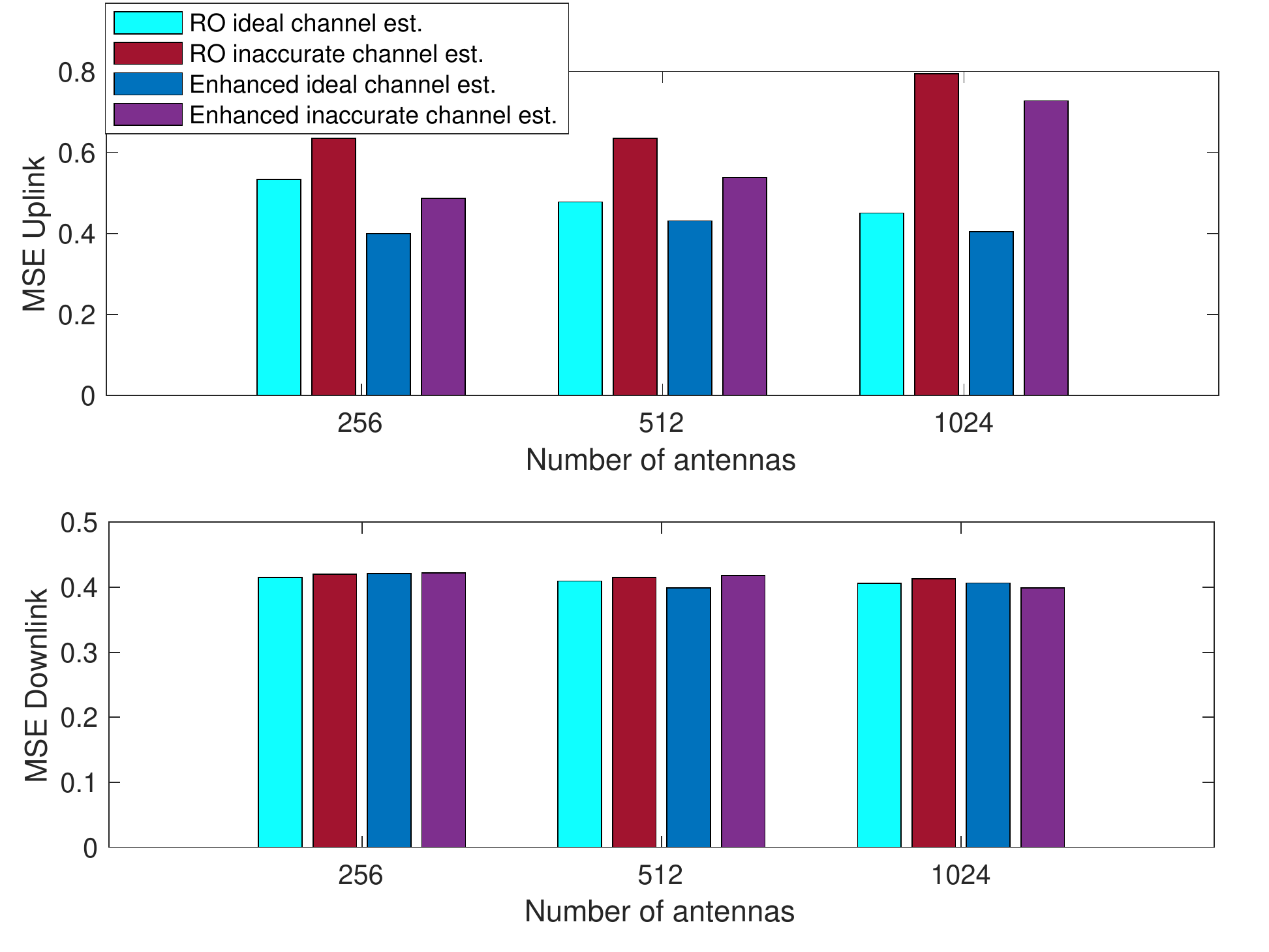}}
    \caption{MSE comparison of the received global ML model parameters when channel hardening and favorable propagation are not fully offered (left) and channel estimation is imperfect (right).}
    \label{fig:corelated}
\end{figure}

\mypara{Robustness.} We now focus on the robustness of the proposed methods, and evaluate the MSEs of the global model parameters obtained at $\ssf{SNR} = 10$ dB through $2000$ Monte Carlo experiments. \zixiangb{
Fig.~\ref{fig:corelated} reports the achieved MSEs of the random orthogonalization method when the (approximate) channel hardening and favorable propagation are not strictly offered, i.e., the wireless channels are correlated.}
We consider two channel correlation models with covariance matrix elements equal to $1$ on the diagonal and equal to $0.01$ or $0.05$ off the diagonal, respectively. \zixiangb{It is observed that} when the off-diagonal elements are $0.01$, random orthogonalization performs nearly identically as \zixiangb{that} in the ideal i.i.d. Rayleigh fading channel case. Even when the off-diagonal elements equal to $0.05$, the achieved \zixiangb{MSEs} only \zixiangb{increase} by less than $1$ dB in the worst case (when $M=256$). \zixiangb{The MSEs become closer to those of the i.i.d. Rayleigh channel cases when $M$ increases, as larger antenna arrays offer higher orthogonality.}

\zixiangb{We next evaluate the performance when the estimation of the summation channel $\hbf_s$ (and $g_k$ in the enhanced method) is imperfect.} The right sub-figure of Fig.~\ref{fig:corelated} compares the MSEs of both proposed methods when the channel estimation is obtained \zixiangb{under} $\ssf{SNR} = 20$ dB. \zixiangb{It reveals that} the downlink communication is more robust than the uplink -- the former achieves nearly identical MSEs as the ideal case even when the channel estimation is inaccurate. For the uplink, an imperfect channel estimation increases the MSEs by $1 \sim 3$ dB depending on the antenna configurations. However, we emphasize again that the FL tasks have the same convergence rate under a constant SINR in the uplink communications (thanks to the model differential transmission).

    

    



\begin{figure}[htb]
    \centering
    \subfigure[MNIST uplink]{\includegraphics[width = 0.32\linewidth]{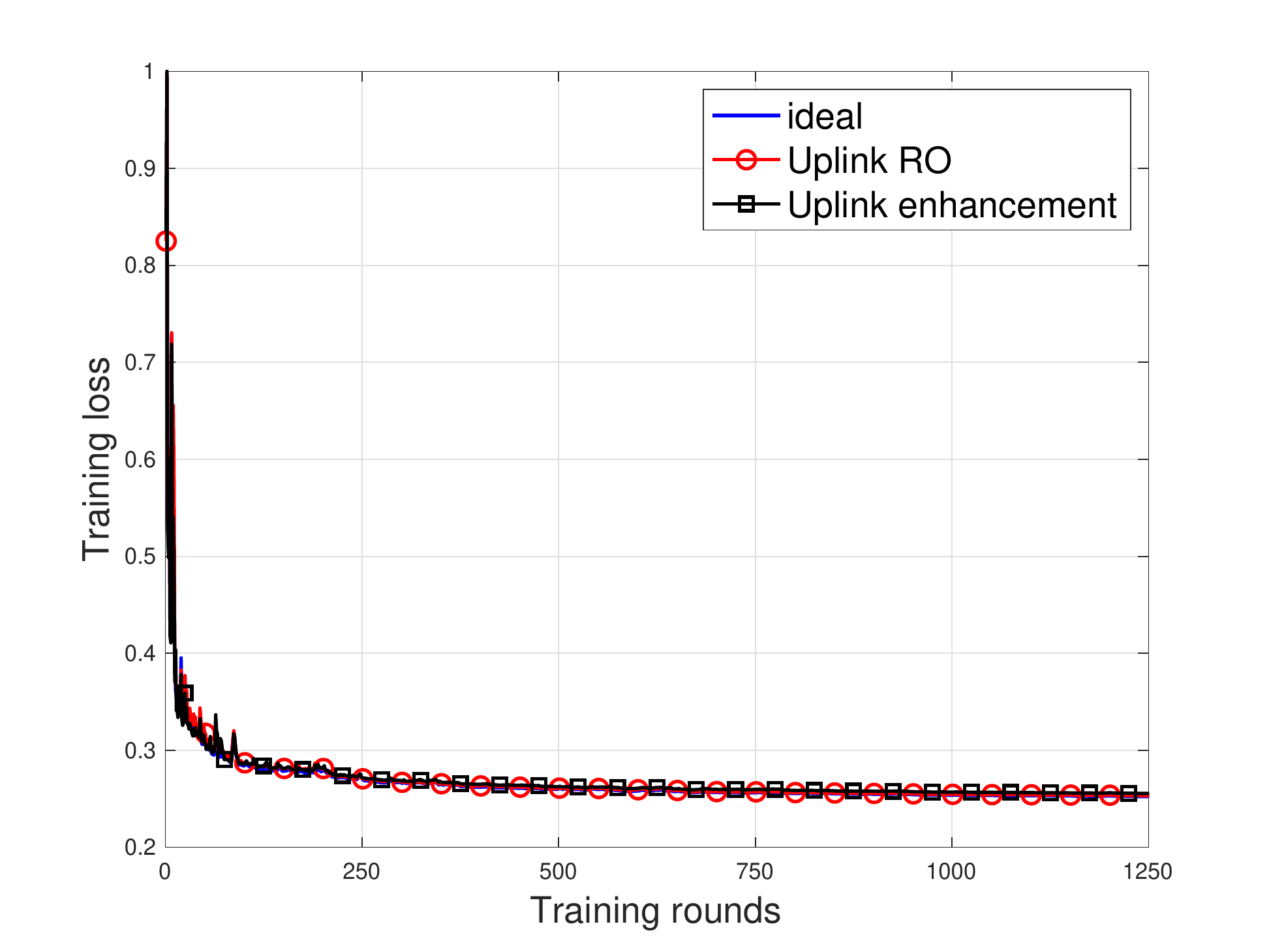}}
    \subfigure[MNIST uplink]{\includegraphics[width = 0.32\linewidth]{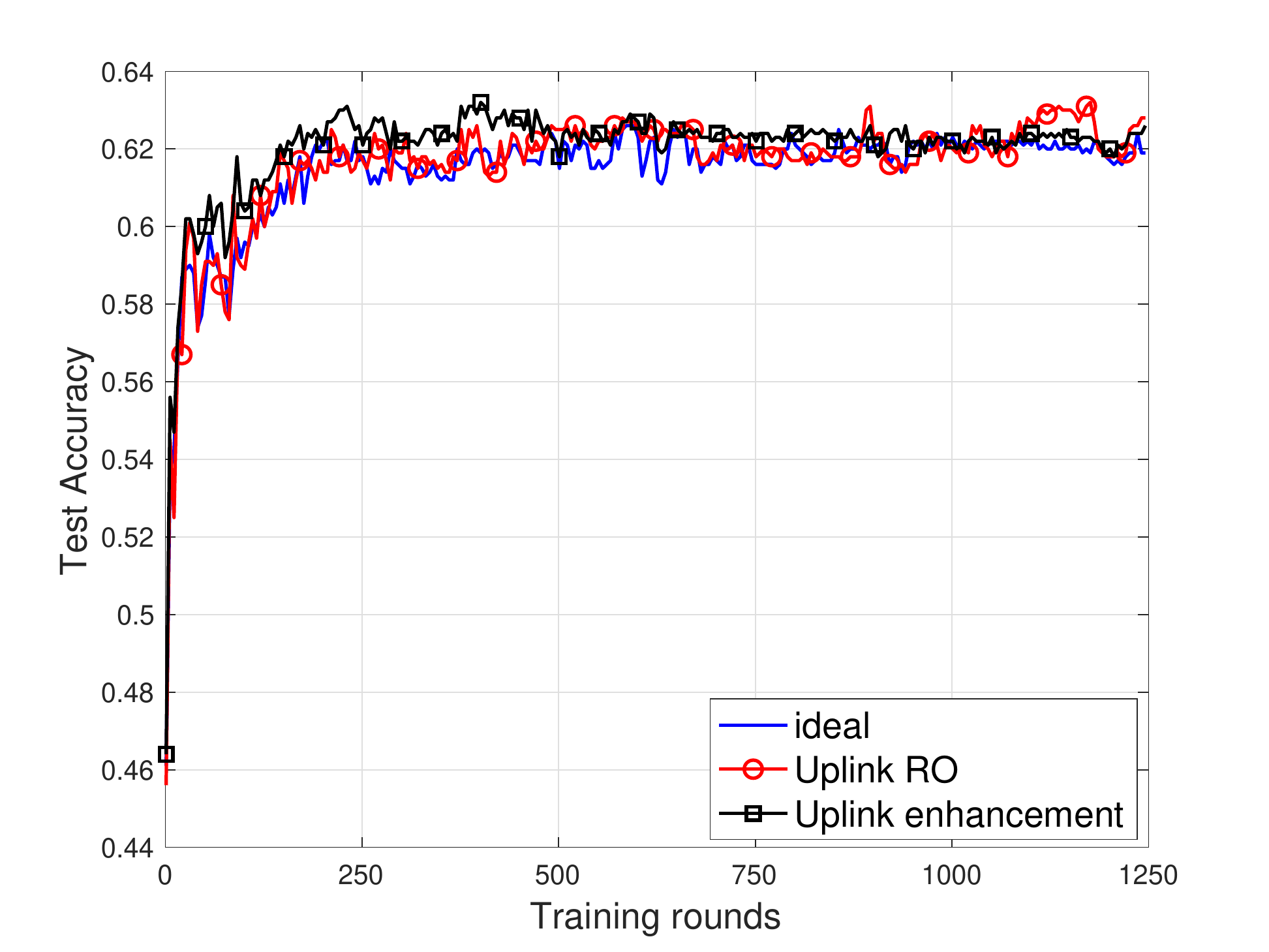}}
    \subfigure[MNIST downlink]{\includegraphics[width = 0.32\linewidth]{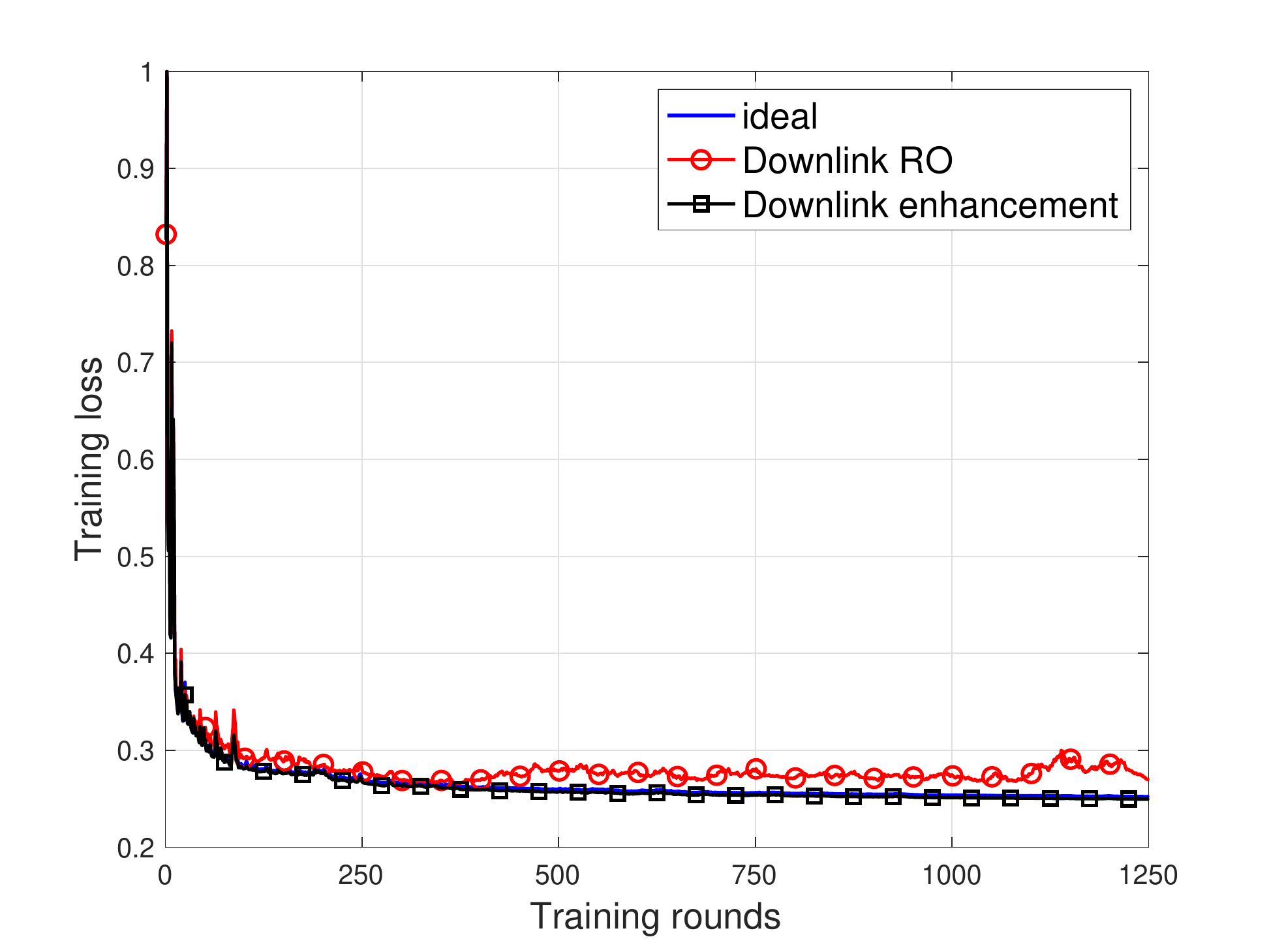}}
    \subfigure[MNIST downlink]{\includegraphics[width = 0.32\linewidth]{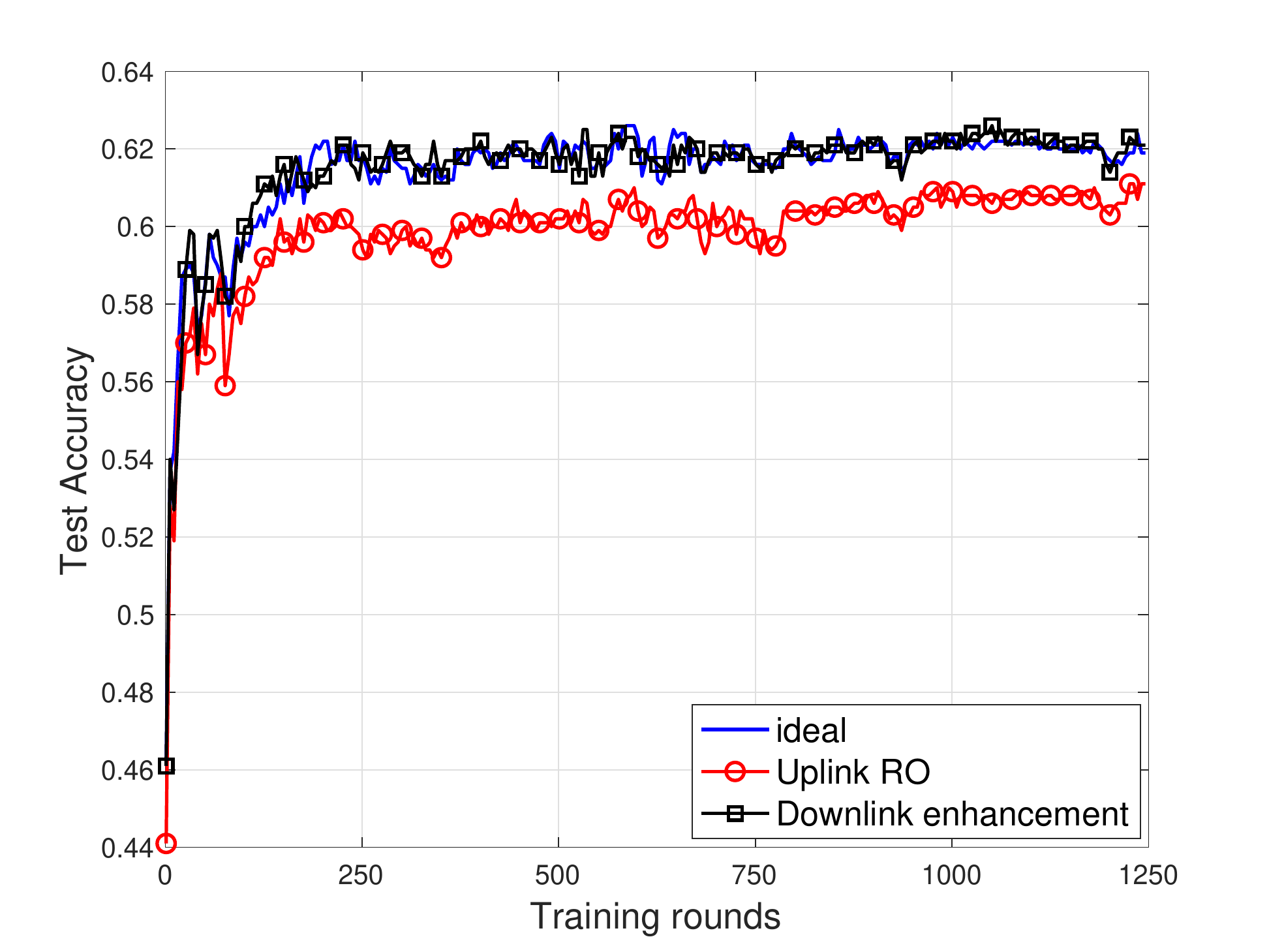}}
    \subfigure[CIFAR-10 uplink+downlink]{\includegraphics[width = 0.32\linewidth]{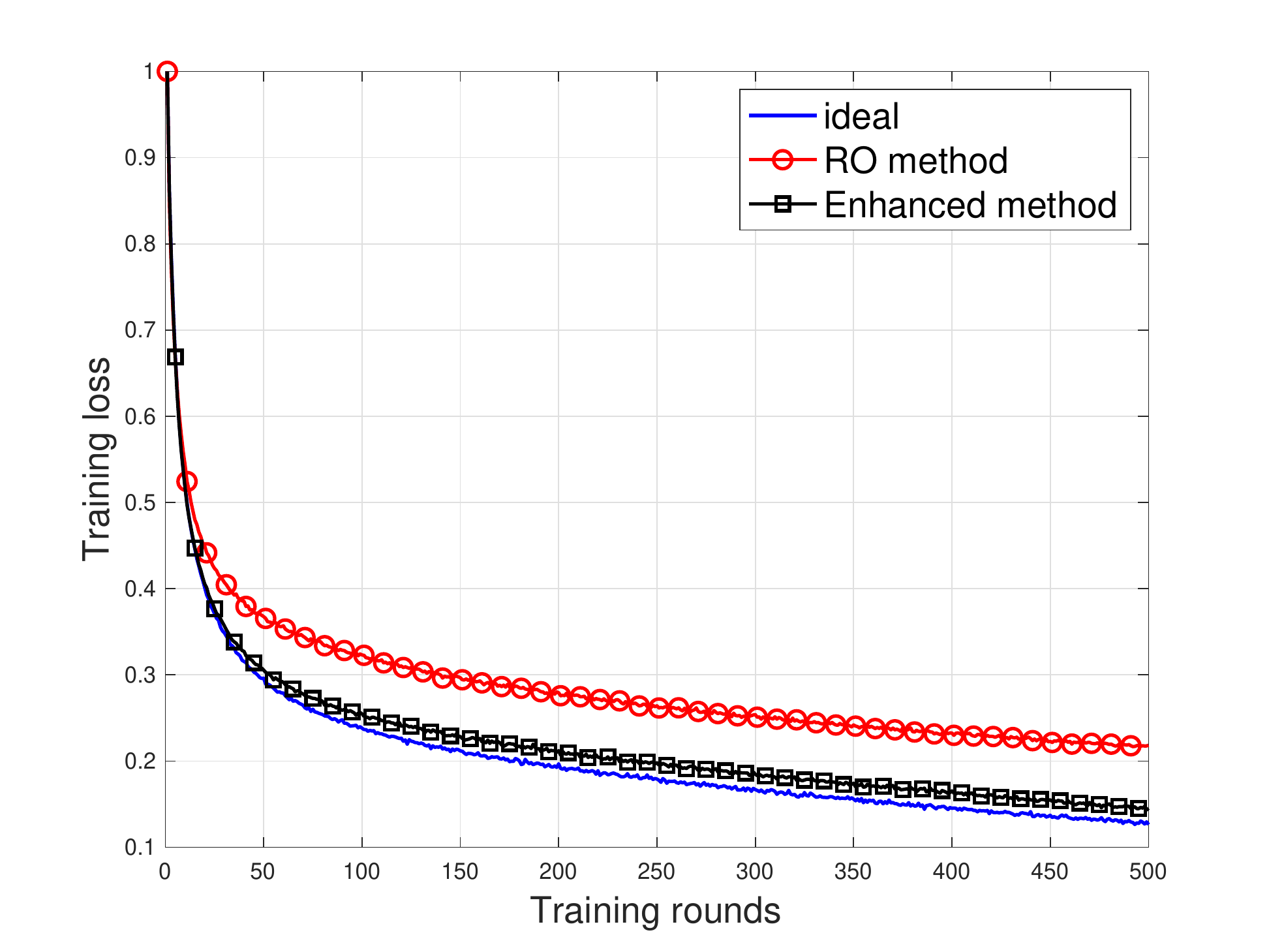}}
    \subfigure[CIFAR-10 uplink+downlink]{\includegraphics[width = 0.32\linewidth]{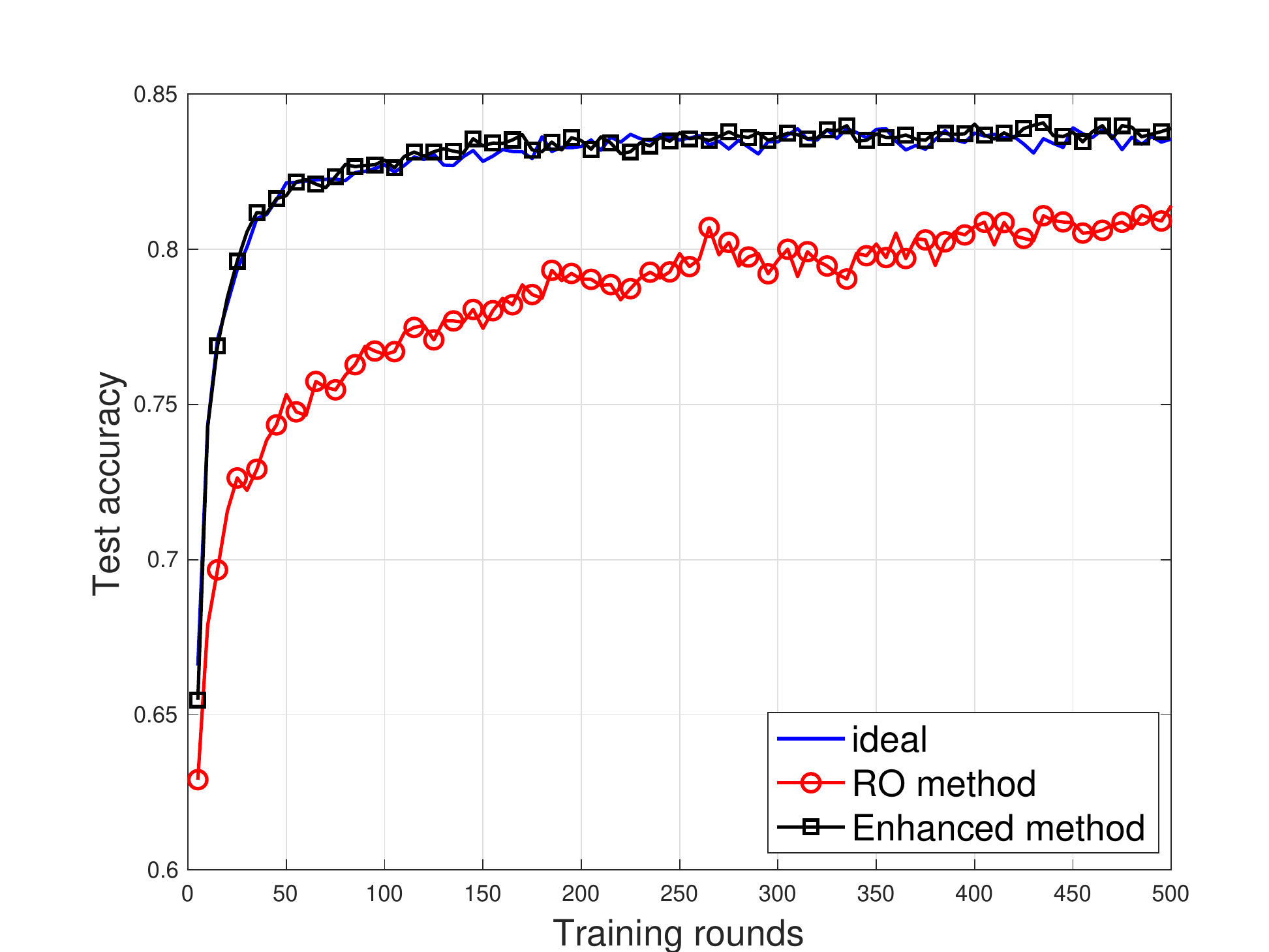}}
    \caption{Comparison of the training loss and test accuracy. (a) and (b): a SVM FL task with an ideal uplink communication (interference and noise free), random orthogonalization,  and the enhanced method; (c) and (d): a SVM FL task with an ideal downlink communication  (interference and noise free), random orthogonalization,  and the enhanced method; (e) and (f): a CIFAR classification FL task with ideal uplink and downlink communications  (interference and noise free), random orthogonalization, and the enhanced method.}
    \label{fig:performanceCompare}
\end{figure}

\subsection{Learning Performance}\label{subsec:learningperformance}
To evaluate the ML performance, we carry out experiments of FL classification tasks using two widely adopted real-world datasets: MNIST and CIFAR-10, via a support vector machine (SVM) model and a convolution neural network (CNN) model, respectively. 

\mypara{MNIST-SVM.} We implement a SVM to classify even and odd numbers in the MNIST handwritten-digit dataset \cite{deng2012mnist}, with $d = 784$. Total clients are set as $N = 20$, the size of each local dataset is $500$, the size of the test set is $2000$, and $E = 1$. The local dataset can be regarded as non-i.i.d. since we only allocate data of one label to each client. We consider a massive MIMO cell with $M = 256$ antennas at the BS and $K = 8$ (out of 20) randomly selected clients are involved in each learning round. The channels between each client and the BS are assumed to be i.i.d. Rayleigh fading.

Fig.~\ref{fig:performanceCompare}(a) and Fig.~\ref{fig:performanceCompare}(b) report the training loss and test accuracy when the uplink adopts the proposed method, and the downlink is assumed to be noise-free. The uplink SNR is set as $10$ dB. We can see that both random orthogonalization and the enhanced method behave almost identically as the ideal case where both uplink and downlink communications are perfect. Note that although the global model received at the BS has noise and interference components, the actual learning performances of the two methods do not deteriorate. Due to the model differential transmission in the uplink communications, the effective SINR of the received global model gradually increases as the model converges, despite the presence of channel interference and noise.  Fig.~\ref{fig:performanceCompare}(c) and Fig.~\ref{fig:performanceCompare}(d) demonstrate the learning performance when the proposed designs are applied to the downlink communications. Since the model differential transmission is infeasible, we set the initial downlink SNR as $0$ dB and scale at a rate of $\mathcal{O}(t^2)$ as the learning progresses (see \cite{wei2021TCCN}). We notice that the learning performance of the enhanced method is almost identical to that of the ideal case, while there exists a performance gap of about $2\%$ test accuracy loss of random orthogonalization. 

\mypara{CIFAR-CNN.} We train a CNN model with two $5 \times 5$ convolution layers (both with 64 channels), two fully connected layers (384 and 192 units respectively) with the $\ssf{ReLU}$ activation, and a final output layer with the softmax activation. The two convolution layers are both followed by $2 \times 2$ max pooling and a local response norm layer. In the FL tasks, we set $N=K=10$, and the size of each local dataset is $1000$, with mini-batch size $50$ and $E=5$. The initial learning rate is $\eta=0.15$ and decays every 10 rounds with rate 0.99. We consider a massive MIMO cell with $M = 1024$ antennas at the BS and the channels between each client and the BS are assumed to be i.i.d. Rayleigh fading. The uplink SNR is set as $10$ dB and the initial downlink SNR is set as $0$ dB, and scales at the rate of $\mathcal{O}(t^2)$.

Fig.~\ref{fig:performanceCompare}(e) and Fig.~\ref{fig:performanceCompare}(f) illustrate the training loss and test accuracy versus the learning rounds when \emph{both} the uplink and downlink communications adopt the random orthogonalization method or the enhanced method, respectively. It is observed that the enhanced method achieves similar training loss and test accuracy as the ideal case. Due to the constant interference in the downlink communications, random orthogonalization incurs a test accuracy loss of about $3\%$.

To summarize, experiments on both datasets demonstrate that random orthogonalization suffers a slight performance degradation over the ideal case when it is applied to the downlink communications. As we have stated in Remark \ref{remark:constantconvergence}, unlike the enhanced method that cancels all interference in the received global model, the interference is constant in the global model obtained via random orthogonalization despite the increased SNR. Note that this gap can be reduced by increasing $M$. Therefore, downlink random orthogonalization is more attractive in systems with large number of antennas or severely limited resources. 

\section{Conclusions}
\label{sec:conc}

Leveraging the unique characteristics of channel hardening and favorable propagation in a massive MIMO system, we have proposed a novel uplink communication method, termed \emph{random orthogonalization}, that significantly reduces the channel estimation overhead while achieving natural over-the-air model aggregation without requiring transmitter side channel state information. We have extended this principle to the downlink communication phase and developed a simple but highly effective model broadcast method for FL. We also relaxed the massive MIMO assumption by proposing an enhanced random orthogonalization design that utilizes channel echos. Theoretical performance analyses, from both communication (CRLB) and machine learning (model convergence rate) perspectives, have been carried out. The theoretical results suggested that random orthogonalization achieves the same convergence rate as vanilla FL with perfect communications asymptotically, and were further validated with numerical experiments. 


\appendices

\section{Preliminaries}
\label{sec:notation}

With a slight abuse of notation, we change the timeline to be with respect to the overall SGD iteration time steps instead of the communication rounds, i.e., $$t=\underbrace{1, \cdots, E}_{\text{round 1}}, \underbrace{E+1, \cdots, 2E}_{\text{round 2}}, \cdots, \cdots, \underbrace{(T-1)E+1, \cdots, TE}_{\text{round $T$}}.$$ 
Note that the global model $\vect{w}_{t}$ is only accessible at the clients for specific $t \in \mathcal{I}_E$, where $\mathcal{I}_E=\{nE ~|~n=1,2,\dots\}$, i.e., the time steps for communication.  The notation for $\eta_t$ is similarly adjusted to this extended timeline, but their values remain constant within the same round. The key technique in the proof is the \emph{perturbed iterate framework}  in \cite{mania2017siam}. In particular, we first define the following local training variables for client $k$: 
\begin{align*}\small
    \vect{v}_{t+1}^k & \triangleq \vect{p}_t^k - \eta_t \nabla \tilde f_k(\vect{p}_t^k); \\
     \vect{u}_{t+1}^k & \triangleq \begin{cases}
        \vect{v}_{t+1}^k & \text{if~} t+1 \notin \mathcal{I}_E, \\
        \frac{1}{K} \sum_{i \in [K]} \vect{v}_{t+1}^i & \text{if~} t+1 \in \mathcal{I}_E; 
    \end{cases} \\
    \vect{w}_{t+1}^k & \triangleq \begin{cases}
         \vect{v}_{t+1}^k & \text{if~} t+1 \notin \mathcal{I}_E, \\
          \frac{1}{K}\sum_{i \in [K]}\hbf_s^H \hbf_i (\vect{v}_{t+1}^i - \vect{w}_{t + 1 -E}) &\\ +  \frac{1}{K}\Nbf_{t + 1} \hbf_s + \vect{w}_{t + 1 - E}
         & \text{if~} t+1 \in \mathcal{I}_E;
    \end{cases}\\
         \vect{p}_{t+1}^k &  \triangleq \begin{cases}
         \vect{v}_{t+1}^k & \text{if~} t+1 \notin \mathcal{I}_E, \\
          \wbf_{t + 1}^k + \tilde\zbf^k_{t + 1}
         & \text{if~} t+1 \in \mathcal{I}_E\text{~and~} k \in [K],
         \\
          \wbf_{t + 1}^k 
         & \text{if~} t+1 \in \mathcal{I}_E \text{~and~} k \notin [K];
    \end{cases}
 \end{align*}
where 
$\Nbf_{t + 1} \triangleq\squab{\nbf_1, \cdots, \nbf_i, \cdots, \nbf_d}^H \in \mathbb{C}^{d\times M}$
is the stack of uplink noise in \eqref{eq:ULcom}, and
$$\tilde\zbf^k_{t + 1} \triangleq
\begin{cases}
        \sqrt{K}\squab{\real(z^k_1/g_1), \cdots, \real(z^k_d/g_d)}^H\in \mathbb{C}^{d\times 1} & \text{if~} k\in[K], \\
        0 & \text{otherwise}, 
    \end{cases}$$
are the downlink noise in \eqref{eq:DLOTAComp}, respectively. Then, we construct the following \textit{virtual sequences}:
$$\avgvect{v}_t=\frac{1}{N}\sum_{k=1}^N \vect{u}_t^k,\;\;\avgvect{u}_t=\frac{1}{N}\sum_{k=1}^N \vect{v}_t^k,\;\;\avgvect{w}_t=\frac{1}{N}\sum_{k=1}^N \vect{w}_t^k,\;\;\text{and~~}
\avgvect{p}_{t} = \frac{1}{N}\sum_{k=1}^N \vect{p}_t^k.$$ 
We also define $\avgvect{g}_{t} = \frac{1}{N}\sum_{k=1}^N \nabla f_k(\vect{w}_t^k)$ and $\vect{g}_{t} = \frac{1}{N}\sum_{k=1}^N \nabla \tilde f_k(\vect{w}_t^k)$ for convenience. Therefore, $\avgvect{v}_{t+1} = \avgvect{w}_t - \eta_t \vect{g}_t$ and $\expt \squab{\vect{g}_t} = \avgvect{g}_t$. Note that the global model $\wbf_{t + 1}$ is only meaningful when $t+1 \in \mathcal{I}_E$, hence we have
$\vect{w}_{t + 1} \triangleq \frac{1}{K}\sum_{k\in[K]} \vect{w}_{t + 1}^k =\frac{1}{N}\sum_{k=1}^N \vect{w}_{t + 1}^k = \avgvect{w}_{t + 1}$.
Thus it is sufficient to analyze the convergence of $\norm{\avgvect{w}_{t+1}-\vect{w}^*}^2$ to evaluate random orthogonalization.

\section{Lemmas}
We first establish the following lemmas that are useful in the proof of Theorem \ref{thm.RO}. 
\label{sec:lemmas}
\begin{lemma}
\label{lemma:one-step-sgd}
Let Assumptions 1-4 hold, $\eta_t$ is non-increasing, and $\eta_t \leq 2 \eta_{t+E}$ for all $t \geq 0$. If $\eta_t \leq {1}/(4L)$, we have
    $\expt \norm{\avgvect{v}_{t+1} - \vect{w}^*}^2  \leq (1-\eta_t \mu) \expt \norm{\avgvect{p}_t - \vect{w}^*}^2 
     + \eta_t^2 \left({\sum_{k\in[N]}H_k^2}/{N^2} + 6 L \Gamma + 8(E-1)^2H^2\right)$.
\end{lemma}

\begin{lemma}
\label{lemma:sample}
    Let Assumptions 1-4 hold. With $\eta_t \leq 2\eta_{t+E}$ for all $t \geq 0$ and $\forall t+1 \in \mathcal{I}_E$, we have
     \begin{equation*}
        \expt \squab{\avgvect{u}_{t+1}} = \avgvect{v}_{t+1}, \;\; \text{and} \;\; \expt \norm{\avgvect{v}_{t+1} - \avgvect{u}_{t+1}}^2 \leq \frac{N-K}{N-1} \frac{4}{K} \eta_t^2 E^2 H^2.
     \end{equation*}
\end{lemma}
Lemmas~\ref{lemma:one-step-sgd} and \ref{lemma:sample} establish bounds for the one-step SGD and random client sampling, respectively. These results only concern the local model update and user selection, and are not impacted by the noisy communication. The proofs are similar to the technique in \cite{stich2018local}, and are omitted due to space limitation.
\begin{lemma}\label{lemma:ROUL}
Let Assumptions 1-4 hold. With $\eta_t \leq 2\eta_{t+E}$ for all $t \geq 0$ and $\forall t+1 \in \mathcal{I}_E$, we have
 \begin{equation*}\label{eqn:RO_unbiased}
         \expt \squab{\avgvect{w}_{t+1}} = \avgvect{u}_{t+1}, \;\; \text{and}\;\; \expt \norm{\avgvect{w}_{t+1} - \avgvect{u}_{t + 1}}^2 \leq \frac{4}{K}\squab{\frac{K}{M} + \frac{1}{\ssf{SNR}_{\ssf{UL}}}} \eta^2_{t} E^2 H^2.
     \end{equation*}
\end{lemma}
\begin{proof}
We take expectation over randomness of fading channel and channel noise. As mentioned in Section \ref{sec:RO}, leveraging channel hardening and favorable propagation properties, we have
\begin{align*}\small
     & \expt \squab{\avgvect{w}_{t+1}}  = \expt\squab{\frac{1}{N}\sum_{k=1}^N \vect{w}_{t + 1}^k} = \expt[\vect{w}_{t + 1}^k]  = \expt\squab{\frac{1}{K}\sum_{i \in [K]}\hbf_s^H \hbf_i (\vect{v}_{t+1}^i - \vect{w}_{t + 1 -E}) +  \frac{1}{K}\Nbf_{t + 1} \hbf_s + \vect{w}_{t + 1 - E}}\\
     & = \expt\squab{\frac{1}{K}\sum_{i \in [K]}\hbf_s^H \hbf_i (\vect{v}_{t+1}^i - \vect{w}_{t + 1 -E})} +  \expt\squab{\frac{1}{K}\Nbf_{t + 1} \hbf_s} + \expt\squab{\vect{w}_{t + 1 - E}}\\
     & = \frac{1}{K}\sum_{i \in [K]}\expt\squab{\sum_{k \in [K]}\hbf_k^H \hbf_i (\vect{v}_{t+1}^i - \vect{w}_{t + 1 -E})} + \vect{w}_{t + 1 - E}\\
     & = \frac{1}{K}\sum_{i \in [K]}\expt\squab{\hbf_i^H \hbf_i (\vect{v}_{t+1}^i - \vect{w}_{t + 1 -E})} + \frac{1}{K}\sum_{i \in [K]}\expt\squab{\sum_{k \in [K], k\neq i}\hbf_k^H \hbf_i (\vect{v}_{t+1}^i - \vect{w}_{t + 1 -E})} + \vect{w}_{t + 1 - E}\\
     & = \frac{1}{K}\sum_{i \in [K]}(\vect{v}_{t+1}^i - \vect{w}_{t + 1 -E}) + \vect{w}_{t + 1 - E} = \frac{1}{K}\sum_{i \in [K]}\vect{v}_{t+1}^i = \avgvect{u}_{t + 1}.
\end{align*}
We next evaluate the variance of $\avgvect{w}_{t + 1}$.
Based on the facts that $\expt[\hbf_i^H\hbf_i] = 1$, and $\forall i\neq j$, we have $\expt[\hbf_i^H\hbf_j] = 0$, $\variance[\hbf_i^H\hbf_j] = \frac{1}{M}$, and $\xbf_i$ and $\xbf_j$ are independent, we have
\begin{align*}\small
\label{eq.A1}
    & \expt \norm{\avgvect{w}_{t+1} - \avgvect{u}_{t + 1}}^2 = \expt\left\Vert \frac{1}{K}\sum_{i \in [K]}\hbf_s^H \hbf_i (\vect{v}_{t+1}^i - \vect{w}_{t + 1 -E}) +  \frac{1}{K}\Nbf_{t + 1} \hbf_s + \vect{w}_{t + 1 - E} - \frac{1}{K}\sum_{i \in [K]}\vect{v}_{t+1}^i\right\Vert^2\\
    & = \expt\norm{\frac{1}{K}\sum_{k\in[K]}\hat\xbf_{k} - \frac{1}{K}\sum_{k\in[K]}\xbf_{k}}^2\\
    & = \frac{1}{K^2}\expt\left\Vert\sum_{k \in [K]} \hbf_k^H \hbf_k \xbf_{k}+ \sum_{k \in [K]}\sum_{j \in [K], j\neq k} \hbf_k^H \hbf_j \xbf_{j} +  \Nbf_{t + 1}\sum_{k \in [K]} \hbf_k -\sum_{k\in[K]}\xbf_{k} \right\Vert^2\\
    & = \frac{1}{K^2}\left[\expt\norm{\sum_{k \in [K]} \hbf_k^H \hbf_k \xbf_{k}}^2 + \expt\norm{\sum_{k \in [K]}\sum_{j \in [K], j\neq k} \hbf_k^H \hbf_j \xbf_{j}}^2 + \expt\norm{\Nbf_{t + 1}\sum_{k \in [K]} \hbf_k}^2 + \expt\norm{\sum_{k\in[K]}\xbf_{k}}^2\right.\\
    & +\left.2\expt\squab{\sum_{k \in [K]} \hbf_k^H \hbf_k \xbf_{k}\sum_{k \in [K]}\sum_{j \in [K], j\neq k} \hbf_k^H \hbf_j \xbf_{j}} + 2\expt\squab{\sum_{k \in [K]} \hbf_k^H \hbf_k \xbf_{k}\Nbf_{t + 1}\sum_{k \in [K]} \hbf_k}\right.\\
    & -\left. 2\expt\squab{\sum_{k \in [K]} \hbf_k^H \hbf_k \xbf_{k}\sum_{k\in[K]}\xbf_{k}} + 2\expt\squab{\sum_{k \in [K]}\sum_{j \in [K], j\neq k} \hbf_k^H \hbf_j \xbf_{j}\Nbf_{t + 1}\sum_{k \in [K]} \hbf_k}\right.\\
    & -\left. 2\expt\squab{\sum_{k \in [K]}\sum_{j \in [K], j\neq k} \hbf_k^H \hbf_j \xbf_{j}\sum_{k\in[K]}\xbf_{k}} -2\expt\squab{\Nbf_{t + 1}\sum_{k \in [K]} \hbf_k \sum_{k\in[K]}\xbf_{k}}\right]\\
    & =\frac{1}{K^2}\squab{\left(1 + \frac{1}{M}\right)\sum_{k \in [K]} \expt\norm{\xbf_{k}}^2 + \frac{K-1}{M}\sum_{k \in [K]} \expt\norm{\xbf_{k}}^2 + \frac{dK}{\ssf{SNR}_{\ssf{UL}}}+\sum_{k \in [K]} \expt\norm{\xbf_{k}}^2 - 2\sum_{k \in [K]}\expt\norm{\xbf_{k}}^2}\\
    & = \frac{1}{K^2}\squab{\frac{K}{M}\sum_{k \in [K]} \expt\norm{\xbf_{k}}^2 + \frac{\sum_{k \in [K]} \expt\norm{\xbf_{k}}^2}{\ssf{SNR}_{\ssf{UL}}}} =\frac{1}{K^2}\squab{\frac{K}{M} + \frac{1}{\ssf{SNR}_{\ssf{UL}}}}\sum_{k \in [K]} \expt\norm{\xbf_{k}}^2 \\
    & \leq \frac{1}{K^2}\squab{\frac{K}{M} + \frac{1}{\ssf{SNR}_{\ssf{UL}}}}\sum_{k \in [K]} E \sum_{i = t + 1 - E}^t\norm{\eta_i \nabla\tilde{f}_k(\wbf_i^k)} \leq \frac{1}{K}\squab{\frac{K}{M} + \frac{1}{\ssf{SNR}_{\ssf{UL}}}} \eta^2_{t+1-E} E^2 H^2 \\
    & \leq \frac{4}{K}\squab{\frac{K}{M} + \frac{1}{\ssf{SNR}_{\ssf{UL}}}} \eta^2_{t} E^2 H^2,
\end{align*}
where in the last inequality we use the fact that $\eta_t$ is non-increasing and $\eta_{t + 1 - E}\leq 2\eta_t$.
\end{proof}

\begin{lemma}\label{lemma:EnhanceDL}
Let Assumptions 1-4 hold and downlink SNR scales $\ssf{SNR}_{\ssf{DL}}\geq \f{1-\mu\eta_t}{\eta^2_t}$ as learning round $t$. $\forall t+1 \in \mathcal{I}_E$, we have
  $\expt \squab{\avgvect{p}_{t+1}} = \avgvect{w}_{t+1}, \; \text{and~} \; \expt \norm{\avgvect{p}_{t+1} - \avgvect{w}_{t + 1}}^2 \leq \left(\frac{dMK}{N^2(K +  M)}\right)\frac{\eta_t^2}{1 - \mu\eta_t}$.
  
\end{lemma}
\begin{proof}
We first show that
    $\expt\squab{\real\left(\frac{z_i^k}{g_k}\right)} = \real\left(\expt\squab{z_i^k}\frac{1}{\expt \squab{ g_k}}\right) = 0$,
and
    $\variance\squab{\real\left(\frac{z_i^k}{g_k}\right)} = \expt\squab{\real\left(\frac{z_i^k}{g_k}\right)\real\left(\frac{z_i^{k^*}}{g_{k^*}}\right)} \\ \expt\squab{\real\left(\frac{z_i^k z_i^{k^*}}{g_k g_{k^*}}\right)}\leq \frac{\expt\squab{\real(z_i^{k^*}z_i^k)}}{\expt\squab{\real(g^*_k g_k)}} = \frac{1/(2\ssf{SNR}_{DL})}{1/2(1 + K/M)} = \left(\frac{M}{K + M}\right)\frac{1}{\ssf{SNR}_{\ssf{DL}}}$,
from which we can easily obtain
    $\expt\squab{{\tilde \zbf_{t + 1}^k}} = \vect{0}$ and $\variance\squab{\tilde \zbf_{t + 1}^k} = \left(\frac{M}{K + M}\right)\frac{d}{\ssf{SNR}_{\ssf{DL}}}$.
Therefore, we have $\expt \squab{\avgvect{p}_{t+1}} = \frac{1}{N}\sum_{k = 1}^N \wbf_{t + 1}^k +\frac{1}{N}\sum_{k\in[K]}\expt\squab{\tilde \zbf_{t + 1}^k} = \avgvect{w}_{t+1}$,
and $\expt \norm{\avgvect{p}_{t+1} - \avgvect{w}_{t + 1}}^2 = \expt\norm{\frac{1}{N}\sum_{k\in[K]}\tilde \zbf_{t + 1}^k}^2 = \frac{1}{N^2}\sum_{k\in[K]}\expt\norm{\tilde \zbf_{t + 1}^k}^2 = \left(\frac{MK}{N^2(K + M)}\right) \frac{d}{\ssf{SNR}_{\ssf{DL}}}\leq \left( \frac{dMK}{N^2(K +  M)}\right)\frac{\eta_t^2}{1 - \mu\eta_t}$.

\end{proof}

\section{Proof of Theorem \ref{thm.RO}}
\label{sec:proofThrm1}
We need to consider four cases for the analysis of the convergence of $\expt\norm{\avgvect{w}_{t+1}-\vect{w}^*}^2$.

1) If $t \notin \mathcal{I}_E$ and $t+1 \notin \mathcal{I}_E$, $\avgvect{v}_{t + 1} = \avgvect{w}_{t + 1}$ and $\avgvect{p}_{t} = \avgvect{w}_{t}$. Using Lemma \ref{lemma:one-step-sgd}, we have:
\begin{equation}
    \label{eqn:case1}
    \expt  \norm{\avgvect{w}_{t+1} - \vect{w}^*}^2 = \expt \norm{\avgvect{v}_{t+1} - \vect{w}^*}^2 \leq (1-\eta_t \mu) \expt \norm{\avgvect{w}_t - \vect{w}^*}^2 + \eta_t^2 \squab{\sum_{k=1}^{N}\frac{H_k^2}{N^2} + 6L \Gamma + 8(E-1)^2 H^2}.
\end{equation}

2) If $t \in \mathcal{I}_E$ and $t+1 \notin \mathcal{I}_E$, we still have $\avgvect{v}_{t+1} = \avgvect{w}_{t+1}$. With $\avgvect{p}_{t} = \avgvect{w}_{t} + \frac{1}{N}\sum_{k=1}^N 
\tilde\zbf_{t}^k$, we have:
\begin{equation*}
    \norm{\avgvect{w}_t - \vect{w}^*}^2 = \norm{\avgvect{p}_t - \avgvect{w}_t + \avgvect{w}_t - \vect{w}^*}^2 = \norm{\avgvect{w}_t - \vect{w}^*}^2 + \underbrace{\norm{\avgvect{w}_t - \avgvect{p}_t}^2}_{A_1} + \underbrace{2\dotp{\avgvect{w}_t - \avgvect{p}_t}{\avgvect{p}_t - \vect{w}^*}}_{A_2}.
\end{equation*}
 We first note that the expectation of $A_2$ over the noise and fading channel randomness is zero since we have $\expt \squab{\avgvect{w}_{t} - \avgvect{p}_{t}} = \vect{0}$. Second, the expectation of $A_1$ can be bounded using Lemma \ref{lemma:EnhanceDL}. We then have
\begin{align}
    & \expt  \norm{\avgvect{w}_{t+1} - \vect{w}^*}^2 = \expt \norm{\avgvect{v}_{t+1} - \vect{w}^*}^2 \nonumber \\
    & \leq (1-\eta_t \mu) \expt \norm{\avgvect{w}_t - \vect{w}^*}^2 +(1-\eta_t \mu) \expt\norm{\avgvect{w}_t - \avgvect{p}_t}^2  + \eta_t^2 \squab{\sum_{k=1}^{N}\frac{H_k^2}{N^2} + 6L \Gamma + 8(E-1)^2 H^2 } \nonumber \\
    & \leq (1-\eta_t \mu) \expt \norm{\avgvect{w}_t - \vect{w}^*}^2 + \eta_t^2 \squab{\sum_{k=1}^{N}\frac{H_k^2}{N^2} + 6L \Gamma + 8(E-1)^2 H^2 + \frac{MK}{N^2(K + M)}}. \label{eqn:case2}
\end{align}
        
3) If $t \notin \mathcal{I}_E$ and $t+1 \in \mathcal{I}_E$, then we still have $\avgvect{p}_{t} = \avgvect{w}_{t}$. For $t+1$, we need to evaluate the convergence of $\expt\norm{\avgvect{w}_{t+1}-\vect{w}^*}^2$. We have
 \begin{equation} \label{eqn:depart0}
    \begin{split}
        \norm{\avgvect{w}_{t+1} - \vect{w}^*}^2 & = \norm{\avgvect{w}_{t+1} - \avgvect{u}_{t+1} + \avgvect{u}_{t+1}- \vect{w}^*}^2 \\
        & = \underbrace{\norm{\avgvect{w}_{t+1} - \avgvect{u}_{t+1}}^2}_{B_1} + \underbrace{\norm{\avgvect{u}_{t+1}- \vect{w}^*}^2}_{B_2} + \underbrace{2\dotp{\avgvect{w}_{t+1} - \avgvect{u}_{t+1}}{\avgvect{u}_{t+1}- \vect{w}^*}}_{B_3}.
    \end{split}
\end{equation}
We first note that the expectation of $B_3$ over the noise is zero since we have $\expt \squab{\avgvect{u}_{t+1} - \avgvect{w}_{t+1}} = \vect{0}$ and the expectation of $B_1$ can be bounded using Lemma \ref{lemma:ROUL}. We next write $B_2$ into
 \begin{equation} \label{eqn:depart1}
    \begin{split}
        \norm{\avgvect{u}_{t+1} - \vect{w}^*}^2 & = \norm{\avgvect{u}_{t+1} - \avgvect{v}_{t+1} + \avgvect{v}_{t+1}- \vect{w}^*}^2 \\
        & = \underbrace{\norm{\avgvect{u}_{t+1} - \avgvect{v}_{t+1}}^2}_{C_1} + \underbrace{\norm{\avgvect{v}_{t+1}- \vect{w}^*}^2}_{C_2} + \underbrace{2\dotp{\avgvect{u}_{t+1} - \avgvect{v}_{t+1}}{\avgvect{v}_{t+1}- \vect{w}^*}}_{C_3}.
    \end{split}
\end{equation}
Similarly, the expectation of $C_3$ over the noise is zero since we have $\expt \squab{\avgvect{u}_{t+1} - \avgvect{v}_{t+1}} = \vect{0}$ and the expectation of $C_1$ can be bounded using Lemma \ref{lemma:sample}. Therefore, we have
   \begin{equation}\label{eqn:case3}
   \begin{split}
        & \expt\norm{\avgvect{w}_{t+1}-\vect{w}^*}^2  \leq \expt\norm{\avgvect{v}_{t+1}- \vect{w}^*}^2 + \frac{4}{K}\squab{\frac{K}{M} + \frac{1}{\ssf{SNR}_{\ssf{UL}}}} \eta^2_{t} E^2 H^2 +  \frac{N-K}{N-1} \frac{4}{K} \eta_t^2 E^2 H^2 \\
        & \leq (1-\eta_t \mu) \expt \norm{\avgvect{w}_t - \vect{w}^*}^2 + \eta_t^2 \left[\sum_{k=1}^N \frac{H_k^2}{N^2} + 6L \Gamma + 8(E-1)^2 H^2 \right.\\
        &\left. + \frac{4}{K}\left(\frac{K}{M} + \frac{1}{\ssf{SNR}_{\ssf{UL}}}\right)  E^2 H^2 +  \frac{N-K}{N-1} \frac{4}{K}  E^2 H^2 \right].  
           \end{split}
      \end{equation}

4) If $t \in \mathcal{I}_E$ and $t+1 \in \mathcal{I}_E$, $\avgvect{v}_{t+1}\neq \avgvect{w}_{t+1}$ and $\avgvect{p}_{t}\neq \avgvect{w}_{t}$. (Note that this is possible only for $E=1$.) Combining the results from the previous two cases, we have 
\begin{equation}
\label{eqn:case4}
\begin{split}\small
    & \expt\norm{\avgvect{w}_{t+1}-\vect{w}^*}^2   \leq (1-\eta_t \mu) \expt \norm{\avgvect{w}_t - \vect{w}^*}^2 \\
    & + \left[\sum_{k=1}^N \frac{H_k^2}{N^2} + 6L \Gamma + 8(E-1)^2 H^2 + \frac{4}{K}\left(\frac{K}{M} + \frac{1}{\ssf{SNR}_{\ssf{UL}}}\right)  E^2 H^2 +  \frac{N-K}{N-1} \frac{4}{K}  E^2 H^2 + \frac{MK}{N^2(K + M)} \right].
 \end{split}
\end{equation}
Let $\Delta_t = \expt \norm{\avgvect{w}_{t}- \vect{w}^*}^2$. From \eqref{eqn:case1}, \eqref{eqn:case2}, \eqref{eqn:case3} and \eqref{eqn:case4}, it is clear that no matter whether $t+1 \in \mathcal{I}_E$ or $t+1 \notin \mathcal{I}_E$, we always have
      $\Delta_{t+1} \leq (1-\eta_t\mu) \Delta_{t} + \eta_t^2 B$,
    where
      $B  = \sum_{k=1}^N \frac{H_k^2}{N^2} + 6L \Gamma + 8(E-1)^2 H^2 + \frac{4}{K}\left(\frac{K}{M} + \frac{1}{\ssf{SNR}_{\ssf{UL}}}\right)  E^2 H^2 +  \frac{N-K}{N-1} \frac{4}{K}  E^2 H^2 + \frac{MK}{N^2(K + M)}$.
 Define $v\triangleq \max\{\frac{4B}{\mu^2}, (1 + \gamma)\Delta_1\}$, by choosing $\eta_t = \frac{2}{\mu(t+\gamma)}$, we can prove $\Delta_t\leq\frac{v}{t + \gamma}$ by induction:
 \begin{equation*}
 \begin{split}
     \Delta_{t + 1} & \leq \left(1 - \frac{2}{t + \gamma}\right)\Delta_t + \frac{4B}{\mu^2(t + \gamma)^2} = \frac{t + \gamma -2}{(t + \gamma)^2}v + \frac{4B}{\mu^2(t + \gamma)^2}\\
     & = \frac{t + \gamma -1}{(t + \gamma)^2}v + \left(\frac{4B}{\mu^2(t + \gamma)^2} - \frac{v}{(t + \gamma)^2}\right) \leq \frac{v}{t + \gamma + 1}.
 \end{split}
 \end{equation*}
 By the $L$-smoothness of $f$ and $v \leq \frac{4B}{\mu^2} + (1 + \gamma)\Delta_1$, we can prove the result in (\ref{eq.convergence}).
 
\bibliographystyle{IEEEtran}
\bibliography{wireless,Shen,ref}

\end{document}

%% file: Backup/com.tex
\newcommand{\ba}{\begin{array}}
\newcommand{\ea}{\end{array}}
\newcommand{\be}{\begin{displaymath}}
\newcommand{\ee}{\end{displaymath}}
\newcommand{\ben}{\begin{equation}}
\newcommand{\een}{\end{equation}}
\newcommand{\bena}{\begin{eqnarray}}
\newcommand{\eena}{\end{eqnarray}}
\newcommand{\beqa}{\begin{eqnarray*}}
\newcommand{\enqa}{\end{eqnarray*}}
\newcommand{\f}{\frac}
\newcommand{\bc}{\begin{center}}
\newcommand{\ec}{\end{center}}
\newcommand{\bi}{\begin{itemize}}
\newcommand{\ei}{\end{itemize}}
\newcommand{\benu}{\begin{enumerate}}
\newcommand{\eenu}{\end{enumerate}}
\newcommand{\bdes}{\begin{description}}
\newcommand{\edes}{\end{description}}
\newcommand{\bt}{\begin{tabular}}
\newcommand{\et}{\end{tabular}}

\newcommand \mubf{\mbox{\boldmath$\mu$\unboldmath}}

\newcommand \fbf{{\bf f}}

\newcommand \hbf{{\bf h}}

\newcommand \nbf{{\bf n}}

\newcommand \vbf{{\bf v}}
\newcommand \wbf{{\bf w}}
\newcommand \xbf{{\bf x}}
\newcommand \ybf{{\bf y}}
\newcommand \zbf{{\bf z}}

\newcommand \Cbf{{\bf C}}

\newcommand \Fbf{{\bf F}}

\newcommand \Hbf{{\bf H}}
\newcommand \Ibf{{\bf I}}

\newcommand \Nbf{{\bf N}}

\newcommand{\circlambda}{\mbox{$\Lambda$
             \kern-.85em\raise1.5ex
             \hbox{$\scriptstyle{\circ}$}}\,}